\newif\ifabstract
\abstracttrue
\newif\iffull
\ifabstract \fullfalse \else \fulltrue \fi

\documentclass[11pt]{article}
\usepackage{amsfonts}
\usepackage{amssymb}
\usepackage{amstext}
\usepackage{amsmath}
\usepackage{xspace}
\usepackage{ntheorem}
\usepackage{graphicx}
\usepackage{url}
\usepackage{graphics}
\usepackage{colordvi}
\usepackage{hyperref}
\usepackage{subfigure}
\usepackage{cite}
\usepackage{xcolor}

\advance \oddsidemargin by -0.8in
\newcommand{\myparskip}{3pt}
\parskip \myparskip
\newcommand{\gset}{\mathcal{U}}
\newcommand{\gkrv}{\ensuremath{\gamma_{\mbox{\tiny{\sc CMG}}}}}
\newcommand{\gKRV}{\ensuremath{\gamma_{\mbox{\tiny{\sc CMG}}}}}

\newcommand{\trset}{\tilde{\mathcal{R}}}

\newcommand{\PoS}{Path-of-Sets System\xspace}
\newcommand{\pos}{\mathbb{P}}
\newcommand{\hpos}{\mathbb{H}}

\newcommand{\ceil}[1]{\ensuremath{\left\lceil#1\right\rceil}}
\newcommand{\floor}[1]{\ensuremath{\left\lfloor#1\right\rfloor}}

\def\etal{et al.\xspace}

\newcommand{\set}[1]{\left\{ #1 \right\}}
\newcommand{\sse}{\subseteq}

\newcommand{\pset}{{\mathcal{P}}}
\newcommand{\tpset}{\tilde{\mathcal{P}}}

\newcommand{\qset}{{\mathcal{Q}}}
\newcommand{\hqset}{\hat{\mathcal{Q}}}
\newcommand{\hrset}{\hat{\mathcal{R}}}
\newcommand{\hM}{\hat{M}}
\newcommand{\hw}{\hat{w}}
\newcommand{\hN}{\hat N}
\newcommand{\hD}{\hat D}
\newcommand{\tqset}{\tilde{\mathcal{Q}}}
\newcommand{\hG}{\hat G}

\newcommand{\bset}{{\mathcal{B}}}

\newcommand{\cset}{{\mathcal{C}}}
\newcommand{\tcset}{\tilde {\mathcal{C}}}

\newcommand{\jset}{{\mathcal{J}}}

\newcommand{\yset}{{\mathcal{Y}}}
\newcommand{\rset}{{\mathcal{R}}}

\newcommand{\sset}{{\mathcal{S}}}

\newcommand{\nots}{\overline S}

\newcommand{\be}{\begin{enumerate}}
\newcommand{\ee}{\end{enumerate}}
\newcommand{\bd}{\begin{description}}
\newcommand{\ed}{\end{description}}
\newcommand{\bi}{\begin{itemize}}
\newcommand{\ei}{\end{itemize}}

\newtheorem{theorem}{Theorem}[section]
\newtheorem{lemma}[theorem]{Lemma}
\newtheorem{observation}[theorem]{Observation}
\newtheorem{corollary}[theorem]{Corollary}
\newtheorem{claim}[theorem]{Claim}

\newtheorem*{definition}{Definition.}

\newenvironment{proof}{\par \smallskip{\bf Proof:}}{\hfill\stopproof}
\def\stopproof{\square}
\def\square{\vbox{\hrule height.2pt\hbox{\vrule width.2pt height5pt \kern5pt
\vrule width.2pt} \hrule height.2pt}}

\newenvironment{proofof}[1]{\noindent{\bf Proof of #1.}}%
{\hfill\stopproof}

\renewcommand{\phi}{\varphi}

\newcommand{\half}{\ensuremath{\frac{1}{2}}}

\newcommand{\poly}{\operatorname{poly}}

\newcommand{\tw}{\mathrm{tw}}

\newenvironment{properties}[2][0]
{
\begin{enumerate} \setcounter{enumi}{#1}}{\end{enumerate}}

\newcommand{\mynote}[1]{{\sc\bf{[#1]}}}

\newcommand{\dmax}{d_{\mbox{\textup{\footnotesize{max}}}}}
\newcommand{\out}{\operatorname{out}}

\newcommand{\Erdos}{Erdos\xspace}
\newcommand{\Posa}{P\'{o}sa\xspace}
\newcommand{\EP}{{\Erdos-\Posa}\xspace}

\begin{document}

\title{Towards Tight(er) Bounds for the Excluded Grid Theorem\footnote{Extended abstract appeared in SODA 2019.}}

\author{Julia Chuzhoy\thanks{Toyota Technological Institute at Chicago. Email: {\tt cjulia@ttic.edu}. Part of the work was done while the author was a Weston visiting professor in the Department of Computer Science and Applied Mathematics, Weizmann Institute. Supported in part by NSF grants CCF-1318242 and CCF-1616584.}\and Zihan Tan\thanks{Computer Science Department, University of Chicago. Email: {\tt zihantan@uchicago.edu}. Supported in part by NSF grants CCF-1318242 and CCF-1616584.}}

\begin{titlepage}
\thispagestyle{empty}
\maketitle
\begin{abstract}
We study the Excluded Grid Theorem, a fundamental structural result in graph theory, that was proved by Robertson and Seymour in their seminal work on graph minors. The theorem states that there is a function $f: \mathbb{Z}^+ \to \mathbb{Z}^+$, such that for every integer $g>0$, every graph of treewidth at least $f(g)$ contains the $(g\times g)$-grid as a minor. For every integer $g>0$, let $f(g)$ be the smallest value for which the theorem holds. Establishing tight bounds on $f(g)$ is an important graph-theoretic question. Robertson and Seymour showed that $f(g)=\Omega(g^2\log g)$ must hold. For a long time, the best known upper bounds on $f(g)$ were super-exponential in $g$. The first polynomial upper bound of $f(g)=O(g^{98}\poly\log g)$ was proved  by Chekuri and Chuzhoy. It was later improved to $f(g) = O(g^{36}\poly \log g)$, and then to $f(g)=O(g^{19}\poly\log g)$. In this paper we further improve this bound to $f(g)=O(g^{9}\poly \log g)$. We believe that our proof is significantly simpler than the proofs of the previous bounds. Moreover, while there are natural barriers that seem to prevent the previous methods from yielding tight bounds for the theorem, it seems conceivable that the techniques proposed in this paper can lead to even tighter bounds on $f(g)$.
\end{abstract}

\thispagestyle{empty}

\end{titlepage}

\section{Introduction}


The Excluded Grid theorem is a fundamental result in graph theory,  that was proved by Robertson and Seymour~\cite{RS-grid} in their Graph Minors series. The theorem states that there is a function $f: \mathbb{Z}^+ \to \mathbb{Z}^+$, such that for every integer $g>0$, every graph of treewidth at least $f(g)$ contains the $(g\times g)$-grid as a minor. The theorem has found many applications in graph theory and algorithms, including routing problems~\cite{robertson1995graph}, fixed-parameter tractability~\cite{DemaineH-survey,DemaineH07}, and \EP-type results \cite{RS-grid,thomassen1988presence, Reed-chapter, fomin2011strengthening}. 
For an integer $g>0$, let $f(g)$ be the smallest value, such that  every graph of treewidth at least $f(g)$ contains the $(g\!\times\!g)$-grid as a minor. An important open question is establishing tight bounds on $f$. Besides being a fundamental graph-theoretic question in its own right, improved upper bounds on $f$ directly affect the running times of numerous algorithms that rely on the theorem, as well as parameters in various graph-theoretic results, such as, for example, \EP-type results.

On the negative side, it is easy to see that $f(g)=\Omega(g^2)$ must hold. Indeed, the complete graph on $g^2$ vertices has treewidth $g^2-1$, while the size of the largest grid minor in it is $(g\times g)$.  Robertson \etal \cite{RobertsonST94} showed a slightly stronger bound of $f(g)=\Omega(g^2\log g)$, by using $\Omega(\log n)$-girth constant-degree expanders, and they conjectured that this bound is tight. Demaine \etal~\cite{demaine2009algorithmic} conjectured that $f(g)=\Theta(g^3)$.

On the positive side, for a long time, the best known upper bounds on $f(g)$ remained super-exponential in $g$: the original bound of \cite{RS-grid} was improved by Robertson, Seymour and Thomas in~\cite{RobertsonST94} to $f(g)=2^{O(g^5)}$. It was further improved to $f(g)=2^{O(g^2/\log g)}$ by Kawarabayashi and Kobayashi \cite{kawarabayashi2012linear} and by Leaf and Seymour \cite{leaf2015tree}. The first polynomial upper bound of $f(g)=O(g^{98}\poly\log g)$ was proved by Chekuri and Chuzhoy~\cite{CC14}.  The proof is constructive and provides a randomized algorithm that, given an $n$-vertex graph $G$ of treewidth $k$, finds a model of the $(g\times g)$-grid minor in $G$, with $g=\tilde\Omega(k^{1/98})$, in time polynomial in both $n$ and $k$. Unfortunately, the proof itself is quite complex.
In a subsequent paper, Chuzhoy~\cite{GMT-STOC} suggested a relatively simple framework for the proof of the theorem, that can be used to obtain a polynomial bound $f(g)=O(g^c)$ for some constant $c$. Using this framework, she obtained an upper bound of $f(g)=O(g^{36}\poly \log g)$, but unfortunately the attempts to optimize the constant in the exponent resulted in a rather technical proof. Combining the ideas from~\cite{CC14} and \cite{GMT-STOC}, the upper bound was further improved to  $f(g) = O(g^{19}\poly \log g)$ in~\cite{chuzhoy2016improved}. We note that the results in~\cite{GMT-STOC} and \cite{chuzhoy2016improved}  are existential. 

The main result of this paper is the proof of the following theorem.
\begin{theorem}
	\label{thm:main}
	There exist constants $c_1,c_2>0$, such that for every integer $g\geq 2$, every graph of treewidth at least $k=c_1g^{9}\log^{c_2} g$ contains the $(g\times g)$-grid as a minor. 
\end{theorem}

Aside from significantly improving the best current upper bounds for the Excluded Grid Theorem, we believe that our framework is significantly simpler than the previous proofs. Even though a relatively simple strategy for proving the Excluded Grid Theorem was suggested in~\cite{GMT-STOC}, this strategy only led to weak polynomial bounds on $f(g)$, and obtaining tighter bounds required technically complex proofs. For example, the best previous bound of   $f(g) = O(g^{19}\poly \log g)$ is an 80-page manuscript. Moreover, there are natural barriers that we discuss below, that prevent the strategy proposed in~\cite{GMT-STOC} from yielding tight bounds on $f(g)$, while it is conceivable that the approach proposed in this paper will lead to even tighter bounds on $f(g)$.

\noindent{\bf Our Techniques.}
We now provide an overview of our techniques, and of the techniques employed in the previous proofs~\cite{CC14,GMT-STOC,chuzhoy2016improved} that achieve polynomial bounds on $f(g)$.

One of the central graph-theoretic notions that we use is that of well-linkedness. Informally, we say that a subset $T$ of vertices of a graph $G$ is well-linked if the vertices of $T$ are, in some sense, well-connected in $G$. Formally, for every pair $T',T''\subseteq T$ of disjoint subsets of $T$ with $|T'|=|T''|$, there must be a collection $\pset$ of paths connecting every vertex of $T'$ to a distinct vertex of $T''$ in $G$, such that the paths in $\pset$ are disjoint in their vertices --- we call such a set $\pset$ of paths a set of \emph{node-disjoint paths}. It is well known that, if $T$ is the largest-cardinality subset of vertices of $G$, such that $T$ is well-linked in $G$, then the treewidth of $G$ is $\Theta(|T|)$  (see e.g.~\cite{Reed-chapter}).

As in the proofs of~\cite{CC14,GMT-STOC,chuzhoy2016improved}, the main combinatorial object that we use is the {\em Path-of-Sets System}, that was introduced in~\cite{CC14}; a somewhat similar object (called a {\em grill}) was also studied by Leaf and Seymour~\cite{leaf2015tree}. 
A \PoS $\pos$ of width $w$ and length $\ell$ (see Figure~\ref{fig:PoS}) consists of a sequence $\cset=(C_1,\ldots,C_{\ell})$ of $\ell$ connected sub-graphs of the input graph $G$ that we call clusters. For each cluster $C_i\subseteq G$, we are given two disjoint subsets $A_i,B_i\subseteq V(C_i)$ of its vertices of cardinality $w$ each. We require that the vertices of $A_i\cup B_i$ are well-linked  in $C_i$\footnote{We use a somewhat weaker property than well-linkedness here, but for clarity of exposition we ignore these technicalities for now.}. Additionally, for each $1\leq i<\ell$, we are given a set $\pset_i$ of $w$ node-disjoint paths, connecting every vertex of $B_i$ to a distinct vertex of $A_{i+1}$. The paths in $\bigcup_i \pset_i$ must be all mutually disjoint, and they cannot contain the vertices of $\bigcup_{i'=1}^{\ell}C_{i'}$ as inner vertices. Chekuri and Chuzhoy~\cite{CC14}, strengthening a similar result of Leaf and Seymour~\cite{leaf2015tree}, showed that, if a graph $G$ contains a \PoS of width $g^2$ and length $g^2$, then $G$ contains an $(\Omega(g)\times \Omega(g))$-grid as a minor. Therefore, in order to prove Theorem~\ref{thm:main}, it is enough to show that a graph of treewidth $\Omega(g^{9}\log^{c_2} g)$ contains a \PoS of width and length $\Omega(g^2)$.

\begin{figure}[h]
\centering
\subfigure[A \PoS]{\scalebox{0.5}{\includegraphics{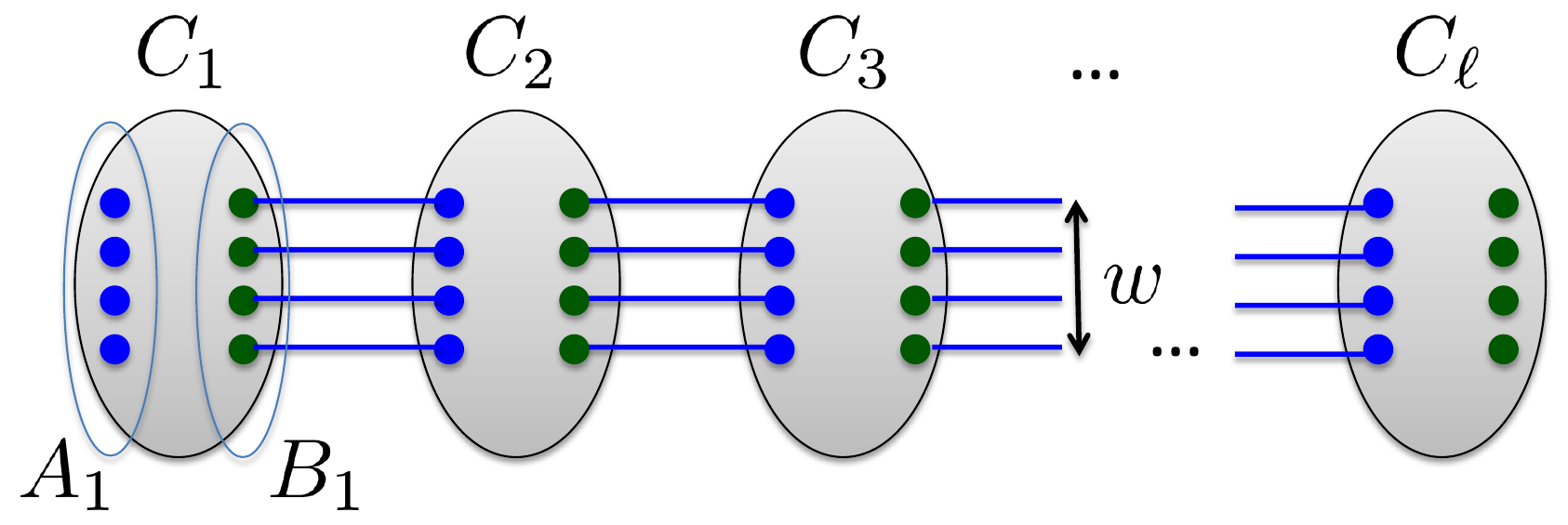}}\label{fig:PoS}}
\hspace{1cm}
\subfigure[A hairy \PoS]{
\scalebox{0.35}{\includegraphics{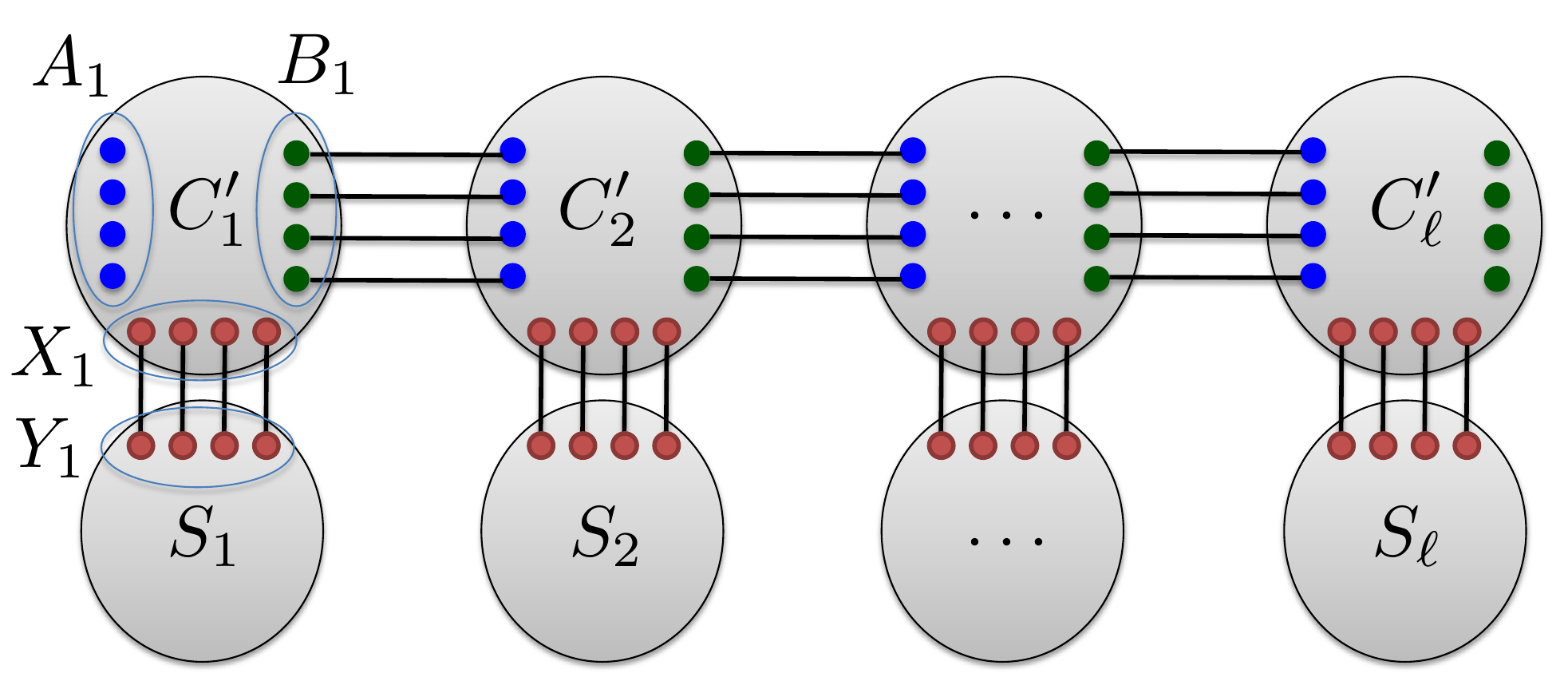}}\label{fig: hairy-PoS}}
\caption{A \PoS and a hairy \PoS \label{fig: two PoS's}}
\end{figure}

Note that, if a graph $G$ has treewidth $k$, then it contains a set $T$ of $\Omega(k)$ vertices, that we call terminals, that are well-linked in $G$.
In~\cite{CC14}, the following approach was employed to construct a large \PoS in a large-treewidth graph. Let $C$ be any connected sub-graph of $G$, and let $\Gamma(C)$ be the set of the boundary vertices of $C$ --- all vertices of $C$ that have a neighbor lying outside of $C$. We say that $C$ is a \emph{good router} if: (i) the vertices of $\Gamma(C)$ are well-linked\footnote{Here, a much weaker definition of well-linkedness was used, but we ignore these technicalities in this informal overview.} in $C$; and (ii) there is a large set of node-disjoint paths connecting vertices of $\Gamma(C)$ to the terminals in $T$. The proof of~\cite{CC14} consists of two steps. First, they show that, if the treewidth of $G$ is large, then $G$ contains a large number of disjoint good routers. In the second step, a large subset of the good routers are combined into a \PoS. Both these steps are quite technical, and rely on many previous results, such as the cut-matching game~\cite{KRV}, graph-reduction step preserving element-connectivity~\cite{element-connectivity,ChekuriK09}, edge-splitting~\cite{edge-connectivity}, and LP-based approximation algorithms for bounded-degree spanning tree~\cite{Singh-Lau}, to name just a few. While the bound on $f(g)$ that this result produces is weak: $f(g)=O(g^{98}\poly\log g)$, this result has several very useful consequences that were exploited in all subsequent proofs of the Excluded Grid Theorem, including the one in the current paper. First, the result implies that for any integer $\ell>0$, a graph of treewidth $k$ contains a \PoS of length $\ell$ and width $\Omega(k/(\poly(\ell\log k)))$. In particular, setting $\ell=\Theta(\log k)$, we can obtain a \PoS of length $\ell$ and width $\Omega(k/\poly\log k)$. This fact was used in~\cite{tw-sparsifiers} to show that any graph $G$ of treewidth $k$ contains a sub-graph $G'$ of treewidth $\Omega(k/\poly\log k)$, whose maximum vertex degree bounded by a constant, where the constant bounding the degree can be made as small as $3$. This latter result proved to be a convenient starting point for subsequent improved bounds for the Excluded Grid Theorem.

In~\cite{GMT-STOC}, a different strategy for obtaining a \PoS was suggested. Recall that, if a graph $G$ has treewidth $k$, then it contains a set $T$ of $\Omega(k)$ vertices, that we call terminals, which are well-linked in $G$. Partitioning the terminals into two equal-cardinality subsets $A_1$ and $B_1$, and letting $C_1=G$, we obtain a \PoS of width $\Omega(k)$ and length $1$. 
The strategy now is to perform a number of iterations, where in every iteration, the length of the current \PoS doubles, while its width decreases by some small constant factor $c$. Since we eventually need to construct a \PoS of length $g^2$, we will need to perform roughly $2\log g$ iterations, eventually obtaining a \PoS of length $g^2$ and width $\Omega(k/c^{2\log g})=\Omega(k/g^{2\log c})$. In order to execute a single iteration, a subroutine is employed, that, given a single cluster of the \PoS, splits this cluster into two. Equivalently, given a \PoS of length $1$ and width $w$, it produces a \PoS of length $2$ and width $w/c$. By iteratively applying this procedure to every cluster of the current \PoS, one obtains a new \PoS, whose length is double the length of the original \PoS, and the width decreases by factor $c$. Recall that the width of the final \PoS that we obtain is $\Omega(k/g^{2\log c})$, and we require that it is at least $\Omega(g^2)$, so $k\geq \Omega(g^{2\log c+2})$ must hold. Therefore, the factor $ c$ that we lose in the splitting of a single cluster is critical for the final bound on $f(g)$, and even if this factor is quite small (which seems very non-trivial to achieve), it seems unlikely that this approach would lead to tight bounds for the Excluded Grid Theorem. The best current bound on the loss parameter $c$ is estimated to be roughly $2^{15}$.
Finally, in~\cite{chuzhoy2016improved}, the ideas from~\cite{CC14} and~\cite{GMT-STOC} are carefully combined to obtain a tighter bound of $f(g)=\tilde O(g^{19})$. We note that both the results of~\cite{GMT-STOC} and~\cite{chuzhoy2016improved} critically require that the maximum vertex degree of the input graph is bounded by a small constant, which can be achieved using the results of~\cite{tw-sparsifiers} and~\cite{CC14}, as discussed above.

Our proof proceeds quite differently. Our starting point is a \PoS of length $\ell=\Theta(\log k)$ and width $w=k/\poly\log k$, where $k$ is the treewidth of the input graph $G$. The \PoS can be constructed, using, e.g., the results of~\cite{CC14}. We then transform it into a structure called a \emph{hairy} \PoS (see Figure~\ref{fig: hairy-PoS}) by further splitting every cluster $C_i$ of the original \PoS into two clusters, $C'_i$ and $S_i$. The clusters $C'_1,\ldots,C'_{\ell}$ are connected into a \PoS as before, albeit with a somewhat smaller width $w/c$ for some constant $c$, and for each $1\leq i\leq \ell$, there is a set $\qset_i$ of $w$ node-disjoint paths, connecting $C'_i$ to $S_i$, that are internally disjoint from both clusters. Let $X_i$ and $Y_i$ denote the sets of endpoints of the paths of $\qset_i$ that belong to $C'_i$ and $S_i$, respectively. We require that $Y_i$ is well-linked in $S_i$, and that $A_i\cup B_i\cup X_i$ is well-linked in $C'_i$. The construction of the hairy \PoS from the original \PoS employs a theorem from~\cite{chuzhoy2016improved}, that allows us to split the clusters of the \PoS appropriately.

The main new combinatorial object that we define is a \emph{crossbar}. Recall that we are interested in showing that the input graph $G$ contains the $(g\times g)$-grid as a minor. Intuitively, a crossbar inside cluster $C'_i$ of a hairy \PoS consists of a set $\pset^*_i$ of $g^2$ disjoint paths, connecting vertices of $A_i$ to vertices of $B_i$; and another set $\qset^*_i$ of $g^2$ disjoint paths, where each path of $\qset^*_i$ connects a distinct path of $\pset^*_i$ to a distinct vertex of $X_i$ (see Figure~\ref{fig: crossbar}). Moreover, we require that the paths of $\qset^*_i$ are internally disjoint from the paths in $\pset^*_i$. The main technical result that we prove is that, if the width $w'$ of the hairy \PoS is sufficiently large, then, for every cluster $C'_i$, either it contains a crossbar, or a minor of $C'_i$ contains a \PoS of length and width $\Omega(g^2)$. If the latter happens in any cluster $C'_i$, then we immediately obtain the $(g\times g)$-grid minor inside $C'_i$. Therefore, we can assume that each cluster $C'_i$ contains a crossbar. We then exploit these crossbars in order to show that the graph $G$ must contain an expander on $\Omega(g^2)$ vertices as a minor, such that the maximum vertex degree in the expander is bounded by $O(\log g)$. 
We can then employ known results to show that such an expander must contain the $(g\times g)$-grid as a minor.

\begin{figure}[h]
\centering
\scalebox{0.4}{\includegraphics{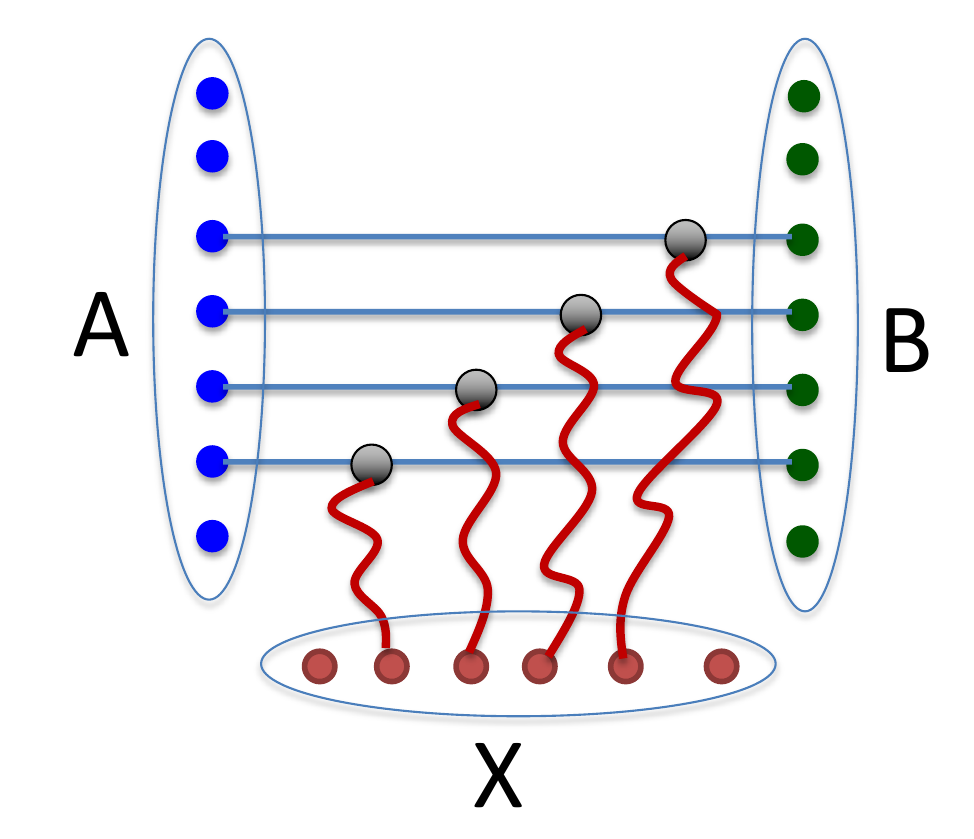}}
\caption{A Crossbar. The paths of $\pset^*$ are shown in blue and the paths of $\qset^*$ in red.\label{fig: crossbar}}
\end{figure}

\noindent{\bf Organization.} We start with preliminaries in Section~\ref{sec:prelim}.  In Section~\ref{sec: Building the Grid} we define a crossbar, state our main structural theorem regarding its existence, and provide the proof of Theorem~\ref{thm:main} using it. 
We provide the proof of the structural theorem in the following two sections: in Section~\ref{sec: Building the Crossbar} we provide a simpler proof of a slightly weaker version of the theorem, 
and in Section~\ref{sec: stronger structural theorem} we provide its full proof. 

\section{Preliminaries}
\label{sec:prelim}

All logarithms in this paper are to the base of $2$. All graphs are finite and they do not have loops. By default, graphs are not allowed to have parallel edges; graphs with parallel edges are explicitly called multi-graphs.

We say that a path $P$ is \emph{disjoint} from a set $U$ of vertices, if $U\cap V(P)=\emptyset$.
We say that it is \emph{internally disjoint} from $U$, if every vertex of $U\cap V(P)$ is an endpoint of $P$. 
Given a set $\pset$ of paths in $G$, we denote by $V(\pset)$ the set of all vertices participating in the paths in $\pset$. 
We say that two paths $P,P'$ are \emph{internally disjoint}, if, for every vertex $v\in V(P)\cap V(P')$, $v$  is an endpoint of both paths.
For two subsets $S,T\subseteq V(G)$ of vertices and a set $\pset$ of paths, we say that $\pset$ \emph{connects} $S$ to $T$ if every path in $\pset$ has one endpoint in $S$ and another in $T$ (or it consists of a single vertex lying in $S\cap T$). We say that a set $\pset$ of paths is \emph{node-disjoint}  iff every pair $P,P'\in \pset$ of distinct paths are disjoint, that is, $V(P)\cap V(P')=\emptyset$. Similarly, we say that a set $\pset$ of paths is \emph{edge-disjoint}  iff for every pair $P,P'\in \pset$ of distinct paths, $E(P)\cap E(P')=\emptyset$.
We sometimes refer to connected subgraphs of a given graph as \emph{clusters}.

\noindent{\bf Treewidth, Minors and Grids.}
The \emph{treewidth} of a graph $G=(V,E)$ is defined via tree-decompositions.  A tree-decomposition of $G$ consists of a tree $\tau$, and, for each node $v\in V(\tau)$, a subset $B_v\subseteq V$ of vertices of $G$ (called a \emph{bag}), such that: (i) for each edge $(v,v') \in E$, there is a node $u \in V(\tau)$ with $v,v' \in B_u$; and (ii) for each vertex $v \in V$, the set $\{u\in V(\tau)\mid v\in B_u\}$ of nodes of $\tau$ induces a non-empty connected subtree of $\tau$. The {\em width} of a tree-decomposition is $\max_{v \in V(\tau)}\set{|B_v|} - 1$, and the \emph{treewidth} of a graph $G$, denoted by $\tw(G)$, is the width of a minimum-width tree-decomposition of $G$.

We say that a graph $H$ is a \emph{minor} of a graph $G$, iff $H$ can be
obtained from $G$ by a sequence of vertex deletion, edge deletion, and edge contraction operations. Equivalently, a graph $H$ is a minor of $G$ iff there is a function $\phi$, mapping each vertex $v\in V(H)$ to a connected subgraph $\phi(v)\subseteq G$, and each edge $e=(u,v)\in E(H)$ to a path $\phi(e)$ in $G$ connecting a vertex of $\phi(u)$ to a vertex of $\phi(v)$, such that:
(i) for all $u,v\in V(H)$, if $u\neq v$, then $\phi(u)\cap \phi(v)=\emptyset$; and
	(ii) the paths in set $\set{\phi(e)\mid e\in E(H)}$ are pairwise internally disjoint, and they are internally disjoint from $\bigcup_{v\in V(H)}\phi(v)$.
A map $\phi$ satisfying these conditions is called a \emph{model}\footnote{Note that this is somewhat different from the standard definition of a model, where  for each edge $e\in E(H)$, $\phi(e)$ is required to be a single edge, but it is easy to see that the two definitions are equivalent.} of $H$ in $G$. We sometimes also say that $\phi$ is an \emph{embedding} of $H$ into $G$, and, for all $v\in V$ and $e\in E$, we specifically refer to $\phi(v)$ as the embedding of  vertex $v$ and to $\phi(e)$ as the embedding of  edge $e$.

The $(g\!\times\!g)$-grid is a graph whose vertex set is: $\set{v(i,j)\mid 1\leq i,j\leq g}$.
The edge set consists of two subsets: a set of \emph{horizontal edges} $E_1=\set{(v(i,j),v(i,j+1))\mid 1\leq i\leq g; 1\leq j<g}$; and a set of \emph{vertical edges} $E_2=\set{(v(i,j),v(i+1,j))\mid 1\leq i<g; 1\leq j\leq g}$. 
We say that a graph $G$ contains the $(g\times g)$-grid minor iff some minor $H$ of $G$ is isomorphic to the $(g \times g)$-grid.

\noindent{\bf Well-linkedness and Linkedness.}

\begin{definition}
Let $G$ be a graph and let $T$ be a subset of its vertices. We say that $T$ is  \emph{node-well-linked} in $G$, iff for every pair  $T',T''$ of disjoint subsets of $T$, there is a set $\pset$ of {\bf node-disjoint} paths in $G$ connecting vertices of $T'$ to vertices of $T''$, with $|\pset|=\min\set{|T'|,|T''|}$. We say that $T$ is \emph{edge-well-linked} in $G$,  iff for every pair  $T',T''$ of disjoint subsets of $T$, there is a set $\pset'$ of {\bf edge-disjoint} paths in $G$ connecting vertices of $T'$ to vertices of $T''$, with $|\pset'|=\min\set{|T'|,|T''|}$. (Note that in the latter definition we allow the paths of $\pset'$ to share their endpoints and inner vertices).

\end{definition}

Even though we do not use it directly, a useful fact to keep in mind is that, if $T$ is the largest-cardinality subset of vertices of a graph $G$, such that $T$ is node-well-linked, then the treewidth of $G$ is between $|T|/4-1$ and $|T|-1$ (see e.g.~\cite{Reed-chapter}).

\begin{definition}
Let $G$ be a graph, and let $A,B$ be two disjoint subsets of its vertices. We say that $(A,B)$ are \emph{node-linked} (or simply \emph{linked}) in $G$, iff for every pair  $A'\subseteq A, B'\subseteq B$ of vertex subsets, there is a set $\pset$ of node-disjoint paths in $G$ connecting vertices of $A'$ to vertices of $B'$, with $|\pset|=\min\set{|A'|,|B'|}$.
\end{definition}

\noindent{\bf A Path-of-Sets System.}
As in previous proofs of the Excluded Grid Theorem, we rely on the notion of \PoS, that we define next (see Figure~\ref{fig:PoS}). 

\begin{definition}
Given integers $\ell,w>0$,
a \PoS $\pos$ of length $\ell$ and width $w$ consists of the following three ingredients: (i) a sequence $\cset=(C_1,\ldots,C_{\ell})$ of mutually disjoint clusters; (ii) for each $1\leq i\leq \ell$, two disjoint subsets $A_i,B_i\subseteq V(C_i)$ of vertices of cardinality $w$ each, and (iii) for each $1\leq i<\ell$, a set $\pset_i$ of $w$ node-disjoint paths connecting $B_i$ to $A_{i+1}$, such that all paths in $\bigcup_{i=1}^{\ell-1}\pset_i$ are node-disjoint, and they are internally disjoint from $\bigcup_{i=1}^{\ell}V(C_i)$. In other words, a path $P\in \pset_i$ starts at a vertex of $B_i$, terminates at a vertex of $A_{i+1}$, and is otherwise disjoint from the clusters in $\cset$.

We say that $\pos$ is a \emph{weak} \PoS iff for each $1\leq i\leq \ell$, $A_i\cup B_i$ is edge-well-linked in $C_i$.
We say that $\pos$ is a \emph{strong} \PoS iff for each $1\leq i\leq \ell$, each of the sets $A_i,B_i$ is node-well-linked in $C_i$, and $(A_i,B_i)$ are linked in $C_i$. 

We sometimes call the vertices of $\bigcup_i(A_i\cup B_i)$ the \emph{nails} of the \PoS.
\end{definition}

Note that a \PoS $\pos$ of length $\ell$ and width $w$ is completely determined by $\cset, \{\pset_i\}_{i=1}^{\ell-1}$, $A_1$ and $B_{\ell}$, so we will denote $\pos=(\cset, \{\pset_i\}_{i=1}^{\ell-1}, A_1, B_{\ell})$. The following theorem was proved in~\cite{CC14}; a similar theorem with slightly weaker bounds was proved in~\cite{leaf2015tree}.

\begin{theorem}
\label{thm: PoS to GM}
There is a constant $c\geq 1$, such that for every integer $g\geq 2$, and for every graph $G$, if $G$ contains a strong \PoS of length $\ell=g^2$ and width  $w=g^2$, then it contains the $(g'\times g')$-grid as a minor, for $g'=\floor{ g/c}$.
\end{theorem}

\noindent{\bf Stitching a \PoS.}
Suppose we are given a \PoS $\pos=(\cset,\set{\pset_i}_{i=1}^{\ell-1},A_1,B_{\ell})$, with $\cset=(C_1,\ldots,C_{\ell})$. Assume that for each odd-indexed cluster $C_{2i-1}$, we select some subsets $A'_{2i-1}\subseteq A_{2i-1}$ ,$B'_{2i-1}\subseteq B_{2i-1}$ of vertices of cardinality $w'$ each, that have some special properties that we desire. We would like to construct a new \PoS, whose clusters are all odd-indexed clusters of $\cset$, and whose nails are $\bigcup_{i=1}^{\ceil{\ell/2}}(A'_{2i-1}\cup B'_{2i-1})$. The stitching procedure allows us to do so, by exploiting the even-indexed clusters of $\pos$ as connectors. The proof of the following claim is straightforward and is deferred to the Appendix.

\begin{claim}\label{claim: stitching}
Let $\pos=(\cset,\set{\pset_i}_{i=1}^{\ell-1},A_1,B_{\ell})$ be a \PoS of length $\ell$ and width $w$ for some $\ell,w\geq 1$. Suppose we are given, for all $1\leq i\leq \ceil{\ell/2}$, subsets  $A'_{2i-1}\subseteq A_{2i-1}$ ,$B'_{2i-1}\subseteq B_{2i-1}$ of vertices of cardinality $w'$ each. Then there is a \PoS  $\hat\pos= (\hat \cset,\set{\hat\pset_i}_{i=1}^{\ceil{\ell/2}-1},\hat A_1,\hat B_{\ceil{\ell/2}})$ of length $\ceil{\ell/2}$ and width $w'$, such that $\hat \cset=(C_1,C_3,\ldots,C_{2\ceil{\ell/2}-1})$; and for each $1\leq i\leq \ceil{\ell/2}$, $\hat A_i=A'_{2i-1}$ and $\hat B_i=B'_{2i-1}$.
\end{claim}

Notice that, if $\pos$ is a strong \PoS in the statement of Claim~\ref{claim: stitching}, then so is $\hat \pos$.

\noindent{\bf Hairy \PoS.} Our starting point is another structure, closely related to the \PoS, that we call a \emph{hairy} \PoS (see Figure~\ref{fig: hairy-PoS}). Intuitively, the hairy \PoS is defined similarly to a strong \PoS, except that now, for each $1\leq i\leq \ell$, we have an additional cluster $S_i$ that connects to $C_i$ with a collection of $w$ node-disjoint paths. We require that the endpoints of these paths are suitably well-linked in $C_i$ and $S_i$, respectively.

\begin{definition}
A hairy \PoS $\hpos$ of length $\ell$ and width $w$ consists of the following four ingredients: 

\begin{itemize}
\item a strong \PoS $\pos=(\cset,\set{\pset_i}_{i=1}^{\ell-1},A_1,B_{\ell})$ of length $\ell$ and width $w$; 
\item  a sequence $\sset=(S_1,\ldots,S_{\ell})$ of disjoint clusters, such that each cluster $S_i$ is disjoint from $\bigcup_{j=1}^{\ell}V(C_j)$ and from $\bigcup_{j=1}^{\ell-1}V(\pset_j)$; 

\item for each $1\leq i\leq \ell$, a set $Y_i\subseteq V(S_i)$ of $w$ vertices that are node-well-linked in $S_i$, and a set $X_i\subseteq V(C_i)$ of $w$ 
vertices, such that $X_i\cap (A_i\cup B_i)=\emptyset$, and $(A_i,X_i)$ are node-linked in $C_i$; and

\item for each $1\leq i\leq \ell$, a collection $\qset_i$ of $w$ node-disjoint paths connecting $X_i$ to $Y_i$, such that all paths in $\bigcup_{j=1}^{\ell}\qset_j$ are disjoint from each other and from the paths in $\bigcup_{j=1}^{\ell-1}\pset_j$, and they are internally disjoint from $\bigcup_{j=1}^{\ell}(S_j\cup C_j)$.
\end{itemize}

\end{definition}

Note that a hairy \PoS $\hpos$ of length $\ell$ and width $w$ is completely determined by $\cset, \sset, \{\pset_i\}_{i=1}^{\ell-1}$, $\{\qset_i\}_{i=1}^{\ell}$, $A_1$ and $B_{\ell}$, so we will denote $\hpos=(\cset, \sset, \{\pset_i\}_{i=1}^{\ell-1}, \{\qset_i\}_{i=1}^{\ell}, A_1, B_{\ell})$. 

We note that Chekuri and Chuzhoy~\cite{CC14} showed that for all integers $\ell,w,k>1$ with $k/\poly\log k=\Omega(w\ell^{48})$, every graph $G$ of treewidth at least $k$ contains a strong \PoS of length $\ell$ and width $w$. 
We prove an analogue of this result for the hairy \PoS.
The proof mostly follows from previous work and is delayed to the Appendix. We will exploit this  theorem only for the setting where $\ell=\Theta(\log k)$ and $w=k/\poly\log k$.

\begin{theorem}
\label{thm: building hairy PoS}
There are constants $c,c'>0$, such that for all integers $\ell,w,k>1$ with $k/\log^{c'}k>cw\ell^{48}$, every graph $G$ of treewidth at least $k$ contains a subgraph $G'$ of maximum vertex degree $3$, such that $G'$ contains a hairy \PoS of length $\ell$ and width $w$.
\end{theorem}

\section{Proof of the Excluded Grid Theorem}
\label{sec: Building the Grid}

In this section, we provide a proof of Theorem~\ref{thm:main}, with some details delayed to Sections~\ref{sec: Building the Crossbar} and~\ref{sec: stronger structural theorem}. 
We start by introducing the main new combinatorial object that we use, called a \emph{crossbar}.

\begin{definition}
Let $H$ be a graph, let $A,B,X$ be three disjoint subsets of its vertices, and let $\rho>0$ be an integer. An $(A,B,X)$-crossbar of width $\rho$ 
consists of a collection $\pset^*$ of $\rho$ paths, each of which connects a vertex of $A$ to a vertex of $B$, and, for each path $P\in \pset^*$, a path $Q_P$, connecting a vertex of $P$ to a vertex of $X$, such that:

\begin{itemize}
\vspace{-3mm}
\item The paths in $\pset^*$ are completely disjoint from each other;
\vspace{-3mm}
\item The paths in $\qset^*=\set{Q_P\mid P\in \pset^*}$ are completely disjoint from each other; and
\vspace{-3mm}
\item For each pair $P\in \pset^*$ and $Q\in \qset^*$ of paths, if $Q\neq Q_P$, then $P$ and $Q$ are disjoint; otherwise $P\cap Q$ contains a single vertex, which is an endpoint of $Q_P$ (see Figure~\ref{fig: crossbar}).\end{itemize}\end{definition}
\vspace{-3mm}

The following theorem is the main technical result of this paper.

\begin{theorem}\label{thm: main for building crossbar}
Let $H$ be a graph and let $g\geq 2$ be an integer, such that $g$ is an integral power of $2$. Let $A,B,X$ be three disjoint sets of vertices of $H$, each of cardinality $\kappa\geq 2^{22}g^9\log g$, such that every vertex in $X$ has degree $1$ in $H$. Assume further that there is a set $\tpset$ of $\kappa$ node-disjoint paths connecting vertices of $A$ to vertices of $B$ in $H$, and a set $\tqset$ of $\kappa$ node-disjoint paths connecting vertices of $A$ to vertices of $X$ in $H$ (but the paths $P\in \tpset$ and $Q\in \tqset$ are not necessarily disjoint). Then, either $H$ contains an $(A,B,X)$-crossbar of width $g^2$, or there is a minor $H'$ of $H$, that contains a strong \PoS of length $\Omega(g^2)$ and width $\Omega(g^2)$.
\end{theorem}

We defer the proof of Theorem~\ref{thm: main for building crossbar} to Sections~\ref{sec: Building the Crossbar} and~\ref{sec: stronger structural theorem}. In Section~\ref{sec: Building the Crossbar} we provide a simpler proof of a slightly weaker version of Theorem~\ref{thm: main for building crossbar}, where we require that $\kappa=\Omega(g^{10}\log g)$, leading to a slightly weaker bound of $f(g)=O(g^{10}\poly\log g)$ for Theorem~\ref{thm:main}. In Section~\ref{sec: stronger structural theorem}, we provide the full proof of Theorem~\ref{thm: main for building crossbar}.

We use the following theorem to complete the proof of Theorem~\ref{thm:main}.

\begin{theorem}\label{thm: grid minor from hairy PoS}
There is a constant $\tilde c$, such that the following holds. Let $G$ be any graph with maximum vertex degree at most $3$, such that $G$ contains a hairy \PoS $\hpos=(\cset, \sset, \{\pset_i\}_{i=1}^{\ell-1}, \set{\qset}_{i=1}^{\ell},A_1,B_{\ell})$ of length $\ell=\tilde c\log g$ and width $\tilde w\geq g^2$, for some integer $g\ge 2$ that is an integral power of $2$. Assume further that for every odd integer $1\leq i\leq \ell$, there is an $(A_i,B_i,Y_i)$-crossbar in graph $C_i\cup \qset_i$ of width $g^2$. Then $G$ contains the $(g'\times g')$-grid as a minor, for $g'=\Omega(g/\log^{1.5}g)$.
\end{theorem}

We first complete the proof of Theorem~\ref{thm:main} using Theorems~\ref{thm: grid minor from hairy PoS} and \ref{thm: main for building crossbar}. Let $G$ be any graph, and let $k$ be its treewidth. Let $g\geq 2$ be an integer, such that $g$ is an integral power of $2$, and such that for some large enough constants $c_1',c_2'$, $k/(\log k)^{c'_2}>c'_1g^9$. We show that $G$ contains a grid minor of size $(\Omega(g/\poly\log g)\times (\Omega(g/\poly\log g))$. 

Let $\ell=\tilde c\log g=O(\log k)$, where $\tilde c$ is the constant from Theorem~\ref{thm: grid minor from hairy PoS}, and let $w=2^{22}g^9\log g=O(g^9\log k)$. By setting the constants $c_1'$ and $c_2'$ in the bound on $k$ appropriately, we can ensure that the conditions of Theorem~\ref{thm: building hairy PoS}, hold for $\ell,w$ and $k$. From Theorem~\ref{thm: building hairy PoS}, there is a subgraph $G'$ of $G$, of maximum vertex degree $3$, such that $G'$ contains a hairy \PoS $\hpos=(\cset, \sset, \{\pset_i\}_{i=1}^{\ell-1}, \set{\qset_i}_{i=1}^{\ell} A_1,B_{\ell})$ of  length $\ell$ and width $w$.

Let $1\leq i\leq \ell$ be an odd integer. Consider the graph $C_i\cup \qset_i$ of the hairy \PoS. Since every pair of the vertex subsets $A_i,B_i,X_i$ are linked in $C_i$, there is a set $\tpset_i$ of $w$ node-disjoint paths connecting vertices of $A_i$ to vertices of $B_i$, and a set $\tqset'_i$ of $w$ node-disjoint paths connecting vertices of $A_i$ to vertices of $X_i$ in $C_i$. By concatenating the paths in $\tqset'_i$ with the paths in $\qset_i$, we obtain a set $\tqset_i$ of $w$ node-disjoint paths connecting vertices of $A_i$ to vertices of $Y_i$. It is also immediate to verify that every vertex in $Y_i$ has degree $1$ in $C_i\cup \qset_i$.  We can therefore apply Theorem~\ref{thm: main for building crossbar} to the graph $C_i\cup \qset_i$, and the sets $A_i,B_i,Y_i$ of its vertices. If, for any odd integer $1\leq i\leq \ell$, the outcome of Theorem~\ref{thm: main for building crossbar} is a strong \PoS of length $\Omega(g^2)$ and width $\Omega(g^2)$ in some minor of $C_i\cup \qset_i$, then from Theorem~\ref{thm: PoS to GM}, graph $G$ contains a grid minor of size $(\Omega(g)\times \Omega(g))$. Therefore, we can assume from now on that for every odd  integer $1\leq i\leq \ell$, the outcome of Theorem~\ref{thm: main for building crossbar} is an  $(A_i,B_i,Y_i)$-crossbar in $C_i\cup \qset_i$ of width $g^2$. But then from Theorem~\ref{thm: grid minor from hairy PoS}, $G$ contains the $(g'\times g')$-grid as a minor, for $g'=\Omega(g/\log^{1.5}g)$. We conclude that there are constants $c_1',c_2'$, such that for all integers $k,g$ with $k/\log^{c_2'}k\geq c_1'g^{9}$, if a graph $G$ has treewidth $k$, then it contains a grid minor of size $(g'\times g')$, where $g'=\Omega(g/\log^5 g)$. (If $g$ is not an integral power of $2$, then we round it up to the closest integral power of $2$ and absorb this additional factor of $2$ in the constants $c'_1$ and $c_2'$). It follows that there are constants $c''_1,c''_2$, such that for all integers $k,g'$ with $k/\log^{c''_2}(k)\geq c''_1(g')^{9}$, if a graph $G$ has treewidth $k$, then it contains a grid minor of size $(g'\times g')$.

One last issue is that we need to replace the $\poly\log k$ in the above bound on $k$ by $\poly\log g$, as in the statement of Theorem~\ref{thm:main}. Assume that we are given a graph $G$ of treewidth $k$, and an integer $g\geq 2$, such that $k\geq c_1g^{9}\log^{c_2}g$ holds, for large enough constants $c_1,c_2$. Let $k'(g)\leq k$ be the smallest integer for which the inequality $k'(g)\geq c_1g^{9}\log^{c_2}g$ holds. Clearly, for some large enough constant $c$ that is independent of $g$, $k'(g)\leq cg^{10}$, and so $\log k'(g)=\Theta(\log g)$. The treewidth of $G$ is at least $k'(g)$, and, by choosing the constants $c_1$ and $c_2$ appropriately, we can guarantee that $k'(g)/\log^{c''_2}(k'(g))\geq c''_1 g^{9}$. From the above arguments, graph $G$ contains the $(g\times g)$-grid as a minor.
It now remains to prove Theorem~\ref{thm: grid minor from hairy PoS}.

\begin{proofof}{Theorem~\ref{thm: grid minor from hairy PoS}}

We assume that $\ell$ is an even integer. For every odd integer $1\leq i\leq \ell$, let $(\tilde \pset^*_i,\tilde \qset^*_i)$ be the $(A_i,B_i,Y_i)$-crossbar of width $g^2$  in $C_i\cup \qset_i$, and let $\tilde A_i\subseteq A_i$, $\tilde B_i\subseteq B_i$ be the sets of endpoints of the paths in $\tilde \pset^*_i$, lying in $A_i$ and $B_i$, respectively, so that $|\tilde A_i|=|\tilde B_i|=g^2$. Using the stitching claim (Claim~\ref{claim: stitching}), we can obtain a strong \PoS $\pos'=(\cset', \{\pset'_i\}_{i=1}^{\ell/2-1}, A'_1,B'_{\ell/2})$ of length $\ell/2$ and width $g^2$, where $\cset'=(C_1',\ldots,C'_{\ell/2})$, and for $1\leq i\leq \ell/2$, $C'_i=C_{2i-1}$.  We denote, for each $1\leq i\leq \ell/2$, $A'_i=\tilde A_{2i-1}$ and $B'_i=\tilde B_{2i-1}$. We also denote $\pset^*_i=\tilde \pset^*_{2i-1}$ and $\qset^*_i=\tilde \qset^*_{2i-1}$. Recall that for each $1\leq i<\ell/2$, the paths in $\pset'_i$ connect the vertices of $B'_i$ to the vertices of $A'_{i+1}$. Combined with the clusters $S_1,S_3,\ldots, S_{\ell-1}$ and the corresponding path sets $\qset_1,\qset_3,\ldots,\qset_{\ell-1}$, we now obtain a hairy \PoS of length $\ell/2$ and width $g^2$. For convenience, abusing the notation, we denote $\ell/2$ by $\ell$, and we denote this new hairy \PoS by $\hpos=(\cset, \sset, \{\pset_i\}_{i=1}^{\ell-1}, \{\qset_i\}_{i=1}^{\ell}, A_1,B_{\ell})$, where for all $1\leq i<\ell$, paths in $\pset_i$ connect $B_i\subseteq V(C_i)$ to $A_{i+1}\subseteq V(C_{i+1})$.

\begin{definition}We say that a multi-graph $H=(V,E)$ is an $\alpha$-expander, iff $\min_{\stackrel{S\sse V:}{|S|\leq|V|/2}}\set{\frac{|E(S,\nots)|}{|S|}}\geq \alpha$.
\end{definition}

We use the following lemma to show that $G$ contains a large enough expander as a minor.

\begin{lemma}\label{lem: model of expander}
	There is a multi-graph $H$, with $|V(H)|=g^2$, and maximum vertex degree $O(\log g)$, such that $H$ is a $1/2$-expander, and it is a minor of $G$. 
\end{lemma}

\begin{proof}	
We use the cut-matching game of Khandekar, Rao and Vazirani~\cite{KRV}, defined as follows. We are given a set $V$ of $N$ nodes, where $N$ is an even integer, and two
players, the cut player and the matching player. 
The game is played in iterations. We start with a graph $H$ with node set $V$ and an empty edge set. In every iteration, some edges are added to $H$. The game ends when $H$ becomes a $\half$-expander. 
The goal of the cut player is to construct a $\half$-expander in as few iterations as
possible, whereas the goal of the matching player is to prevent the construction of the expander for as long as possible. The iterations proceed as follows.  In every iteration $j$, the cut player chooses a partition $(Z_j, Z'_j)$ of $V$ with $|Z_j| = |Z'_j|$, and the
matching player chooses a perfect matching $M_j$ that matches the nodes of $Z_j$ to the nodes of $Z'_j$.  The edges of $M_j$ are then
added to $H$.  Khandekar, Rao, and Vazirani \cite{KRV} showed that there is a strategy for the cut player that guarantees that after
$O(\log^2{N})$ iterations the graph $H$ is a $(1/2)$-expander. Orecchia \etal \cite{OrecchiaSVV08} strengthened
this result by showing that after $O(\log^2{N})$ iterations the graph $H$ is an $\Omega(\log{N})$-expander.  We use another strengthening of the result of \cite{KRV}, due to Khandekar \etal \cite{logarithmic-CMG}.


\begin{theorem} [\cite{logarithmic-CMG}]
	\label{thm: cut-matching game}
	There is a constant $\gkrv$, and an adaptive strategy for the cut player such that,
	no matter how the matching player plays, after $\ceil{\gkrv \log N}$ iterations, multi-graph $H$ is a $1/2$-expander with high probability.
\end{theorem}

We note that the strategy of the cut player is a randomized algorithm that computes a partition $(Z_j,Z'_j)$  of $V$ to be used in the subsequent iteration. It is an adaptive strategy, in that the partition depends on the previous partitions and on the responses of the matching player in previous iterations.
Note that the maximum vertex degree in the resulting expander is exactly $\ceil{\gkrv \log N}$, since the set of its edges is a union of $\ceil{\gkrv \log N}$ matchings.

We assume that the constant $\tilde c$ in the statement of Theorem~\ref{thm: grid minor from hairy PoS} is at least $8\gkrv$, so that the length $\ell$ of the hairy \PoS $\hpos$ is at least $\ceil{\gkrv \log (g^2)}$. We construct a $\half$-expander $H$ over a set $U$ of $g^2$ vertices, and embed it into $G$, as follows.

Let $\rset$ be the set of $g^2$ node-disjoint paths, obtained by concatenating the paths of $\pset^*_1,\pset_1,\ldots,\pset_{\ell-1},\pset^*_{\ell}$.	
For each vertex $u\in U$, we let $\phi(u)$ be any path of $\rset$, that we denote by $R_u$, such that for $u\neq u'$, $R_u\neq R_{u'}$. Next, we will construct the set of edges of $H$ over the course of $\ceil{\gkrv \log (g^2)}$ iterations, by running the cut-matching game. Recall that in each iteration $1\leq j\leq \ceil{\gkrv \log (g^2)}$, the cut player computes a partition $(Z_j,Z_j')$ of $U$ into two equal-cardinality subsets. Our goal is to compute a perfect matching $M_j$ between $Z_j$ and $Z'_j$, whose edges will be added to $H$. For each edge $e\in M_j$, we will also compute its embedding $\phi(e)$ into $G$, such that $\phi(e)$ is contained in $ \qset^*_{j}\cup S_{j}$, and for all $e\neq e'$, $\phi(e)\cap \phi(e')=\emptyset$.

We now fix some index $1\leq j\leq \ceil{\gkrv \log (g^2)}$, and consider iteration $j$ of the Cut-Matching game. 
Then for every vertex $v\in U$, there is a path $Q_j(v)\in \qset^*_j$, that connects a vertex of $R_v$ to some vertex of $Y_{j}$, that we denote by $y_j(v)$. We think of $y_j(v)$ as the \emph{representative} of vertex $v$ for iteration $j$. Note that the paths of $\qset^*_j$ are internally disjoint from the paths of $\rset$.

Consider the partition $(Z_j,Z'_j)$ of $U$ computed by the cut player. This partition naturally defines a partition $(\hat Z_j,\hat Z'_j)$ of $Y_{j}$ into two equal-cardinality subsets, where for each vertex $v\in U$, if $v\in Z_j$, then $y_j(v)$ is added to $\hat Z_j$, and otherwise it is added to $\hat Z'_j$ (see Figure~\ref{fig: concatenation3}). Since the vertices of $Y_{j}$ are node-well-linked in $S_{j}$, there is a set $\qset''_j$ of $|U|/2$ node-disjoint paths in $S_{j}$, connecting vertices of $\hat Z_j$ to vertices of $\hat Z'_j$. The set $\qset''_j$ of paths defines a matching $M_j$ between $Z_j$ and $Z'_j$: for each path $Q\in \qset''_j$, if $y_j(v),y_j(v')$ are the endpoints of $Q$, then we add the edge $e=(v,v')$ to $M_j$, and we let the embedding $\phi(e)$ of this edge be the concatenation of the paths $Q_j(v),Q,Q_j(v')$. It is easy to verify that $\phi(e)$ connects a vertex of $R_v$ to a vertex of $R_{v'}$, and it is internally disjoint from the paths in $\rset$. Moreover, for $e\neq e'$, $\phi(e)\cap \phi(e')=\emptyset$. From Theorem~\ref{thm: cut-matching game}, after $\ceil{\gkrv \log (g^2)}$ iterations, graph $H$ becomes a $\half$-expander, and we obtain a model of $H$ in $G$.

\begin{figure}[h]
\centering
\scalebox{0.5}{\includegraphics[scale=0.8]{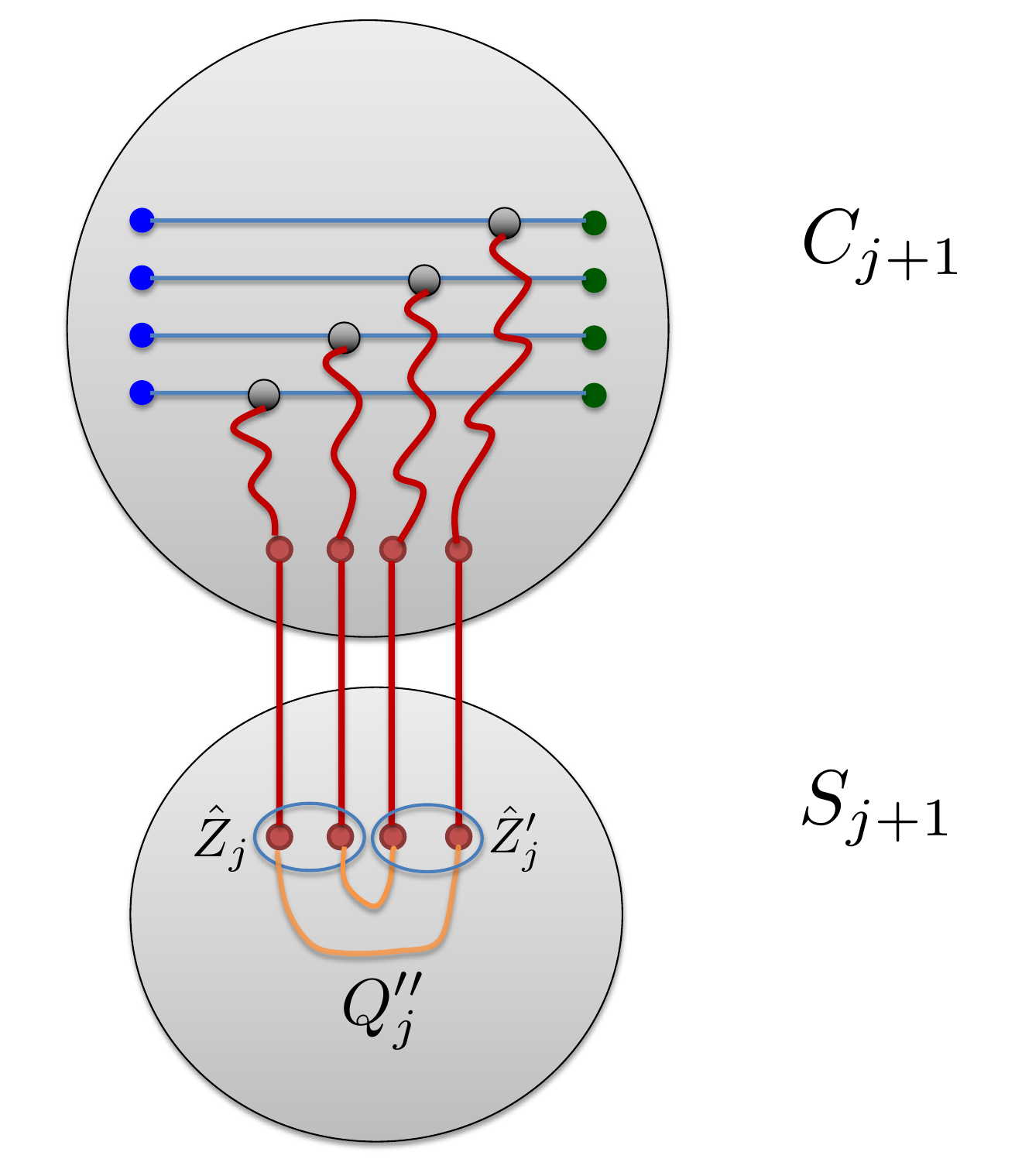}}
\caption{The execution of the $j$th iteration of the cut-matching game.\label{fig: concatenation3}}
\end{figure}

\end{proof}

Let $H$ be the  from Lemma~\ref{lem: model of expander}. 
Recall that $H$ is a $1/2$-expander on $g^2$ vertices, and that its maximum vertex degree is bounded by $O(\log g)$.

 In general, it is well-known that a large enough expander contains a large enough grid as a minor. For example, Kleinberg and Rubinfeld~\cite{KleinbergR} show that a bounded-degree $n$-vertex expander contains any graph with at most $n/\log^{c} n$ edges as a minor, for some constant $c$. Unfortunately, our expander is not bounded-degree, but has degree $O(\log n)$ (where $n$ denotes the number of vertices in the expander). 
 Instead, we will use a recent result of Krivelevich and Nenadov (see Theorem 8.1 in~\cite{expander-minor}). We note that a similar, but somewhat weaker result was independently proved in~\cite{minors-in-expanders}.
 
 Before we state the result, we need the following definition from~\cite{expander-minor}.
 
 \begin{definition}
 Let $G=(V,E)$ be a graph on $n$ vertices. A vertex set $S\subseteq V$ is a \emph{separator} in $G$ if there is a partition $V=(A\cup B\cup S)$ of $V$, such that $|A|,|B|\leq 2n/3$, and $G$ has no edges between $A$ and $B$.
 \end{definition}
 
 The following theorem follows from Theorem 8.1 in~\cite{expander-minor} and the subsequent discussion.
 
\begin{theorem}\label{thm: minor in expander} There exist constants $c,n_0$, such that the following holds. Let $G$ be a graph on $n\geq n_0$ vertices, $0<\alpha<1$ a parameter, and let $H$ be a graph with at most $cn\alpha^2/\log n$ vertices and edges. Then $G$ has a separator of size at most $\alpha n$, or $G$ contains $H$ as a minor.
\end{theorem}

Let $N=g^2$ denote the number of vertices in the $1/2$-expander $H$, and let $d=O(\log g)$ be the maximum vertex degree in $H$. We claim that every separator in $H$ has cardinality at least $N/(24d)$. Indeed, assume otherwise, and let $S$ be a separator of cardinality less than $N/(24d)$ in $H$. Let $(A,B,S)$ be the corresponding partition of $V(H)$, so that $|A|,|B|\leq 2N/3$, and $H$ has no edges between $A$ and $B$. Assume w.l.o.g. that $|A|\leq |B|$. Since $H$ is a $\half$-expander, at least $|A|/2$ edges are leaving $A$, and each such edge must have an endpoint in $S$. Since maximum vertex degree in $H$ is bounded by $d$, $|S|\geq |A|/(2d)$ must hold. However, since we have assumed that $|S|\leq N/(24d)$, and $|B |\leq 2N/3$, we get that $|A|\geq N/6$ must hold, and so $|S|\geq N/(12d)$, a contradiction.
 
 We can now use Theorem~\ref{thm: minor in expander} in graph $H$, with $\alpha=1/(24d)$, to conclude that for some constant $c'$, every graph $H'$ with at most $c'g^2/\log^3 g$  edges and vertices is a minor of $H$. In particular, a grid of size $(g'\times g')$, where $g'=g\cdot \sqrt{c'/\log^3g}$ is a minor of $H$, and hence of $G$.
 \end{proofof}

\section{Building the Crossbar}
\label{sec: Building the Crossbar}

In this section we provide a proof of a weaker version of Theorem~\ref{thm: main for building crossbar}, where we require that $\kappa\ge 2^{22}g^{10}\log g$. We note that this weaker version of Theorem~\ref{thm: main for building crossbar} immediately implies a slightly weaker bound of $f(g)=O(g^{10}\poly\log g)$ for Theorem~\ref{thm:main}. Although this version of Theorem~\ref{thm: main for building crossbar} is slightly weaker, its proof is simpler and it provides a clean framework that includes all the main ideas and the technical machinery needed for the full proof of Theorem~\ref{thm: main for building crossbar}. 
We complete the proof of Theorem~\ref{thm: main for building crossbar} in Section~\ref{sec: stronger structural theorem}.

Recall that we are given a graph $H$, and three disjoint subsets $A,B,X$ of its vertices, each of cardinality $\kappa\geq 2^{22}g^{10}\log g$, such that every vertex of $X$ has degree $1$ in $H$.
We are also given a set $\tpset$ of $\kappa$ node-disjoint paths connecting vertices of $A$ to vertices of $B$, and a set $\tqset$ of $\kappa$ node-disjoint paths connecting vertices of $A$ to vertices of $X$. Our goal is to prove that either $H$ contains an $(A,B,X)$-crossbar of width $g^2$, or  that its minor contains a strong Path-of-Sets system whose length and width are both at least $\Omega(g^2)$. Recall that an $(A,B,X)$-crossbar of width $g^2$ consists of a set $\pset^*$ of $g^2$ node-disjoint paths connecting vertices of $A$ to vertices of $B$, and another set $\qset^*$ of $g^2$ node-disjoint paths, such that the paths in $\qset^*$ are internally disjoint from the paths in $\pset^*$, and for each path $P\in \pset^*$, there is a path $Q_P\in \qset^*$, whose first vertex lies on $P$ and last vertex belongs to $X$. 

It is easy to see that such a crossbar does not always exist. For instance, suppose $H$ is a union of a grid of size $((\kappa+2)\times (\kappa+2))$ and another set $X$ of $\kappa$ vertices, where every vertex of $X$ connects to a distinct vertex of the first column of the grid. Let $A$ contain all vertices of the first row of the grid, excluding the grid corners, and let $B$ contain all vertices of the last row of the grid, excluding the grid corners. The set $\tpset$ of paths is a subset of the columns of the grid, and the existence of the set $\tqset$ of paths is easy to verify. However, there is no $(A,B,X)$-crossbar of width greater than $1$ in this graph. We will show that in such a case, we can find a \PoS of length and width at least $\Omega(g^2)$ in a minor of $H$. 

Our first step is to construct two sets of paths: a set $\pset$ of $\kappa$ node-disjoint paths connecting every vertex of $A$ to a distinct vertex of $B$, and a set $\qset$ of $\kappa$ node-disjoint paths, connecting every vertex of $A$ to a distinct vertex of $X$. Such two sets of paths are guaranteed to exist, as we can use $\pset=\tpset$ and $\qset=\tqset$. 
Let $H(\pset,\qset)=\bigcup_{P\in \pset\cup\qset}P$ be the graph obtained by the union of these paths.
Among all such pairs $(\pset,\qset)$ of path sets, we select the sets $\pset,\qset$ that minimize the number of edges in $H(\pset,\qset)$. For each path $P\in \pset$, we denote by $Q_P\in \qset$ the path originating at the same vertex of $A$ as $P$.
Even though the graph is undirected, it is convenient to think of the paths in $\pset\cup \qset$ as directed away from $A$.

The remainder of the proof consists of six steps. In the first step, we define a new structure that we call a \emph{pseudo-grid}. Informally, a pseudo-grid of depth $D$ consists of a collection $\rset_1,\ldots,\rset_D$ of disjoint subsets of paths in $\pset$ (that is, for all $1\leq i\leq D$, $\rset_i\subseteq \pset$), such that for all $i$, $|\rset_i|\leq g^2$. Additionally, if we denote $\pset'=\pset\setminus\bigcup_i\rset_i$, then there must be a large subset $\qset'\subseteq \set{Q_P\mid P\in \pset'}$ of paths, such that, for all $1\leq i\leq D$, every path $Q\in \qset'$ intersects at least one path of $\rset_i$. We show that, either $H$ contains an $(A,B,X)$-crossbar of width $g^2$, or it contains a pseudo-grid of a large enough depth. 

In the second step, we \emph{slice} this pseudo-grid into a large enough number $M$ of smaller pseudo-grids. Specifically, for each
path $R\in \rset$, we define a sequence $\sigma_1(R),\ldots,\sigma_M(R)$ of disjoint sub-paths of $R$, that appear on $R$ in this order. Let $\Sigma_i=\set{\sigma_i(R)\mid R\in \rset}$. For all $1\leq i\leq M$, we let $\qset_i\subseteq \qset'$ contain only those paths $Q$, for which all vertices of $Q\cap V(\rset)$ belong to $V(\Sigma_i)$. We perform the slicing in a way that ensures that for all $i$, $|\qset_i|$ is large enough. 

In general, for all $1\leq i\leq M$, there are many intersections between the paths in $\qset_i$ and the paths in $\Sigma_i$. But it is possible that some paths $R\in \Sigma_i$ intersect few paths of $\qset_i$ and vice versa. Our third step is a clean-up step, in which we discard all such paths, so that eventually, each path $R\in \Sigma_i$ intersects a large number of paths of $\qset_i$ and vice versa.

In the fourth step, we create clusters that will be used to construct the final \PoS. Specifically, for each $1\leq i\leq M$, we show that there is some cluster $C_i$ in the graph obtained from the union of the paths in $\Sigma_i$ and $\qset_i$, such that there is a large enough collection $\Sigma'_i\subseteq \Sigma_i$ of paths, each of which is contained in $C_i$, and moreover, the endpoints of the paths in $\Sigma'_i$ are well-linked in $C_i$. This step uses standard well-linked decomposition, though its analysis is somewhat subtle.

In the fifth step, we exploit the paths in $\rset$ in order to select a subset of the clusters $C_i$ and link them into a \PoS. Unfortunately, we will only be able to guarantee that, for each cluster $C_i$, the resulting vertex set $A_i\cup B_i$ is edge-well-linked in $C_i$; recall that such a \PoS is called a {weak} \PoS. We then turn it into a strong \PoS using standard techniques in our last step. 

We now provide a formal proof of a weaker version of Theorem~\ref{thm: main for building crossbar}, where we assume that $\kappa\ge 2^{22}g^{10}\log g$, using the sets $\pset,\qset$ of paths that we have defined above.

\subsection{Step 1: Pseudo-Grid}

We define a pseudo-grid, one of our central combinatorial objects.

\begin{definition} Let $D>0$ be an integer.
A pseudo-grid of depth $D$ consists of the following two ingredients.
The first ingredient is a family $\set{\rset_1,\rset_2,\ldots,\rset_D}$ of subsets of $\pset$, where for all $1\leq i\leq D$, $|\rset_i|\leq g^2$, and for all $1\leq i\neq j\leq D$, $\rset_i\cap \rset_j=\emptyset$.
 Let $\rset=\bigcup_{i=1}^D\rset_i$, and let $\pset'=\pset \setminus \rset$.
 The second ingredient is a set $\qset'$ of $\ceil{\kappa/4}$ disjoint paths, where each path $Q\in \qset'$ is a sub-path of a distinct path of $\set{Q_P\mid P\in \pset'}$ (so in particular, $|\pset'|\geq |\qset'|= \ceil{\kappa/4}$), and exactly one endpoint of $Q$ lies in $X$.
Additionally, the following two properties must hold:

\begin{properties}{P}
\item The paths in $\pset'$ are completely disjoint from the paths in $\qset'$; and \label{prop: disjointness of P and Q}
\item For every $1\leq i\leq D$, the number of paths $Q\in \qset'$ with $Q\cap (\bigcup_{P\in \rset_i}P)=\emptyset$ is at most $2g^2$. In other words, all but at most $2g^2$ paths of $\qset'$ must intersect some path of $\rset_i$. \label{prop: intersect each layer} 
 \end{properties}
\end{definition}

The main result of this subsection is the following theorem. 

\begin{theorem}\label{thm: pseudo-grid} Let $D$ be any integer with $1\leq D\leq \kappa/(2g^2)$. Then
either $H$ contains an $(A,B,X)$-crossbar of width $g^2$, or it contains a pseudo-grid of depth $D$.
\end{theorem}

We note that the theorem holds for any value of $\kappa$, and its proof does not rely on specific bounds on $\kappa$.

\begin{proof}
From the definition of the sets $\pset,\qset$ of paths, since every vertex of $X$ has degree $1$, we can assume that no path in $\pset$ contains a vertex of $X$, and that every path of $\qset$ contains exactly one vertex of $X$, that serves as its endpoint.

We perform $D$ iterations, where in iteration $i$ we either construct a crossbar of width $g^2$, or we compute the path set $\rset_i$ of the pseudo-grid. For each $1\leq i\leq D$, we will  denote by $\pset'_i=\pset\setminus \left(\rset_1\cup\cdots\cup\rset_i\right)$ the collection of the remaining paths of $\pset$. We will ensure that for all $i$, $|\rset_i|\leq g^2$. We let $\pset'_0=\pset$.

We now describe the $i$th iteration of our algorithm, whose input is a set $\pset'_{i-1}\subseteq \pset$ of at least $\kappa-(i-1)g^2$ paths.
In order to execute the $i$th iteration, we build a graph $H_i$, that is obtained from the graph $H$, by contracting every path $P\in \pset'_{i-1}$ into a single vertex $v_P$. We keep parallel edges but discard loops. Let $S_i=\set{v_P\mid P\in \pset'_{i-1}}$ be the resulting set of vertices corresponding to the contracted paths. We now compute the largest set $\hat \qset$ of node-disjoint paths in $H_i$, connecting the vertices of $S_i$ to the vertices of $X$. We consider two cases.

\paragraph{Case 1.} The first case happens if $\hat\qset$ contains at least $g^2$ paths.

 In this case, we show that we can construct  an $(A,B,X)$-crossbar of width $g^2$. Consider some path $Q\in \hat\qset$. We can assume without loss of generality that $Q$ contains exactly one vertex of $S_i$, that serves as one of its endpoints. Let $u(Q)$ be this vertex. If $\hat \qset$ contains more than $g^2$ paths, we discard paths from $\hat \qset$ arbitrarily, until $|\hat \qset|=g^2$ holds. We then define $\pset^*$ to be the set of all paths $P\in \pset'_{i-1}$, such that $v_P=u(Q)$ for some path $Q\in \hat \qset$. Finally, we define the set $\qset^*$ of paths of the crossbar, as follows. For each path $Q\in \hat \qset$ in graph $H_i$, we will define a corresponding path $Q'$ in graph $H$, and we will set $\qset^*=\set{Q'\mid Q\in \hat \qset}$. Consider now some path $Q\in \hat \qset$, and let $P\in \pset^*$ be the path with $v_P=u(Q)$. Recall that every vertex of $Q$ is either a vertex of $H$, or it is a vertex of the form $v_{P'}$ for some path $P'\in \pset'_{i-1}$. Let $U'(Q)$ be the set of all vertices of $Q$ that belong to $H$, and let $U''(Q)$ be the set of all vertices lying on the paths $P'\in \pset'_{i-1}$, such that $v_{P'}\in V(Q)$. Finally, let $U(Q)=U'(Q)\cup U''(Q)$. Notice that for any two paths $Q,Q'\in \hat Q$, if $Q\neq Q'$, then $U(Q)\cap U(Q')=\emptyset$, as the two paths are node-disjoint. Let $H(Q)$ be the sub-graph of $H$ induced by the vertices of $U(Q)$. Then $H(Q)$ is a connected graph, that contains at least one vertex of $P$ and at least one vertex of $X$. We let $Q'$ be any path in $H(Q)$ connecting a vertex of $P$ to a vertex of $X$, such that $Q'$ is internally disjoint from $P$. Setting $\qset^*=\set{Q'\mid Q\in \hat \qset}$, we now obtain  an $(A,B,X)$-crossbar $(\pset^*,\qset^*)$ of width $g^2$. Indeed, from the above discussion, it is immediate that the paths of $\pset^*$ are mutually node-disjoint, and so are the paths of $\qset^*$. From our construction, each path $Q'\in \qset^*$ connects a distinct path $P\in \pset^*$ to a vertex of $X$. Consider now any pair $P\in \pset^*,Q'\in \qset^*$ of such paths. If $v_P=u(Q')$, then from our construction $P\cap Q'$ consists of a single vertex, that serves as an endpoint of $Q'$. Otherwise, $Q'\cap P=\emptyset$: indeed, if $\hat Q\in \hat\qset$ is the path with $u(\hat Q)=v_{P}$, then, since $Q'$ and $\hat Q$ are disjoint from each other, $Q'$ may not contain the vertex $v_P$, and so $Q'$ may not contain any vertex of $P$.

\paragraph{Case 2.} We now assume that $\hat \qset$ contains fewer than $g^2$ paths. From Menger's theorem, there is a set $J_i$ of at most $g^2$ vertices in graph $ H_i$, such that in $ H_i\setminus J_i$, there is no path connecting a vertex of $S_i$ to a vertex of $X$. Note that $J_i$ may contain vertices of $S_i\cup X$. 

We partition $J_i$ into two subsets: $J_i'=J_i\cap S_i$, and $J_i''=J_i\setminus S_i$. Notice that each vertex in $J_i''$ is also a vertex in the original graph $H$. We then let $\rset_i\subseteq \pset'_{i-1}$ be the set of all paths $P\in \pset'_{i-1}$, whose corresponding vertex $v_P\in J_i'$. Clearly, $|\rset_i|\leq |J_i|\leq g^2$. We define $\pset'_i=\pset\setminus(\rset_1\cup\cdots\cup \rset_i)=\pset'_{i-1}\setminus\rset_i$.
Let $V_i=J''_i\cup(\bigcup_{P\in \rset_i}V(P))$, a set of vertices of the original graph $H$, and let $\qset'_i=\set{Q_P\mid P\in \pset'_i}$. Then each path in $\qset'_i$ must contain a vertex of $V_i$. For each such path $Q\in \qset'_i$, let $v_i(Q)$ be the last vertex of $Q$ that belongs to $V_i$, and let $\sigma_i(Q)$ be the sub-path of $Q$ between $v_i(Q)$ and the endpoint of $Q$ that belongs to $X$. Note that, as $|J''_i|\leq g^2$, for all but at most $g^2$ paths $Q\in \qset'_i$, the vertex $v_i(Q)$ lies on some path of $\rset_i$. We call such a path $Q\in \qset'_i$ an \emph{$i$-good path}. We will use the following immediate observation:

\begin{observation}\label{obs: disj}
For each path $Q\in \qset'_i$, the segment $\sigma_i(Q)$ cannot contain any vertex of $\bigcup_{P\in \pset'_i}V(P)$.
\end{observation}

We continue this process for $D$ iterations. If we did not construct an $(A,B,X)$-crossbar of width $g^2$, then we obtain the sets $\rset_1,\ldots,\rset_D\subseteq\pset$ of paths, where for all $i$, $|\rset_i|\leq g^2$, and we set $\pset'=\pset'_D$. Clearly, for all $1\leq i\neq j\leq D$, $\rset_i\cap \rset_j=\emptyset$. Since $D\leq \kappa/(2g^2)$, we get that $|\pset'|\geq \kappa/4$.  We let $\qset'=\set{\sigma_D(Q)\mid Q\in \qset'_D}$. If $|\qset'|>\ceil{\kappa/4}$, then we discard arbitrary paths from $\qset'$ until $|\qset'|=\ceil{\kappa/4}$ holds.
We now claim that $\set{\rset_1,\ldots,\rset_D}$ and $\qset'$ define a pseudo-grid of depth $D$. Indeed, as already observed, for each $1\leq i\leq D$, $|\rset_i|\leq g^2$, and for all $1\leq i\neq j\leq D$, $\rset_i\cap \rset_j=\emptyset$.
From the definition of the segments $\sigma_D(Q)$ for $Q\in \qset'_D$, every path in $\qset'$ has exactly one endpoint in $X$.
 From Observation~\ref{obs: disj}, the paths in $\qset'$ are disjoint from the paths in $\pset'$, thus establishing Property~(\ref{prop: disjointness of P and Q}).

It now remains to establish Property~(\ref{prop: intersect each layer}).
Consider some index $1\leq i\leq D$, and some path $Q\in \qset'$, that is both $i$-good and $D$-good. Consider the two corresponding segments $\sigma_i(Q)$ and $\sigma_D(Q)$. Recall that, since $Q$ is $i$-good, $\sigma_i(Q)$ intersects some path of $\rset_i$, but, since $\rset_D\subseteq \pset'_i$, from Observation~\ref{obs: disj},  $\sigma_i(Q)$ cannot contain a vertex of $\bigcup_{P\in \rset_D}V(P)$. As $\sigma_D(Q)$ is $D$-good, it contains a vertex of $\bigcup_{P\in \rset_D}V(P)$. We conclude that $\sigma_i(Q)\subseteq\sigma_D(Q)$, and so $\sigma_D(Q)$ intersects some path of $\rset_i$. As at most $g^2$ paths of $\qset'$ are not $i$-good, and at most $g^2$ paths are not $D$-good,  all but at most $2g^2$ paths of $\qset'$ must intersect some path of $\rset_i$.
 \end{proof}

 We apply Theorem~\ref{thm: pseudo-grid} to graph $H$ with the depth parameter $D=64g^4$. If the outcome is an $(A,B,X)$-crossbar of width $g^2$, then we return this crossbar and terminate the algorithm. Therefore, we assume from now on that the outcome of the theorem is a pseudo-grid of  depth $D$. 

\subsection{Step 2: Slicing the Paths of $\rset$}
\label{subsec:slicing}
 Recall that for each $1\leq i\leq D$, there are at most $2g^2$ paths $Q\in\qset'$, such that $Q$ does not intersect any path of $\rset_i$. We discard all such paths from $\qset'$, obtaining a set $\qset''\subseteq\qset'$ of paths. Observe that we discard at most $2g^2 D=128g^6<\kappa/8$ paths, and, since $\qset'=\ceil{\kappa/4}$, we get that $|\qset''|\geq \kappa/8$.
 We now have the following property:
 
 \begin{properties}{I}
\item For each path $Q\in \qset''$, for every $1\leq i\leq D$, path $Q$ intersects at least one path of $\rset_i$.  \label{prop: intersection}

 \end{properties}
 
 
 We denote  $|\rset|=N$, so $D\leq N\leq Dg^2$. Let $A'\subseteq A$ and $B'\subseteq B$ be the sets of endpoints of the paths of $\rset$ lying in $A$ and $B$, respectively. 
 
 Let $H'$ be the sub-graph of $H$, obtained by taking the union of all paths in $\rset$ and all paths in $\qset''$.  
The next observation follows from the definition of the sets $\pset$ and $\qset$ of paths.

\begin{observation}\label{obs: irrelevant edge} Let $e$ be any edge of $H'$ lying on any path in $\rset$, such that $e$ does not lie on any path $Q$ of the original set $\qset$. Then the largest number of node-disjoint paths connecting vertices of $A'$ to vertices of $B'$ in $H'\setminus\set{e}$ is at most $N-1$.
\end{observation}

\begin{proof}
The key is that the paths in $\qset'$ are disjoint from the paths of $\pset'=\pset\setminus \rset$, and so the paths in $\pset'$ are disjoint from the graph $H'$. Assume for contradiction that $H'\setminus\set{e}$ contains a set $\pset^*$ of $N$ node-disjoint paths connecting vertices of $A'$ to vertices of $B'$. Then we could have defined $\pset$ as $\pset'\cup \pset^*$, and the resulting graph obtained by the union of the paths in $\pset$ and $\qset$ would have contained strictly fewer edges than the corresponding graph with the original definition of $\pset$, contradicting the definition of the sets $\pset,\qset$ of paths.
\end{proof}

 We need the following definition.

\begin{definition}
Given a graph $\hat H$, and two subsets $Y,Z$ of its vertices, with $|Y|=|Z|=r$ for some integer $r$, we say that $\hat H$ has the \emph{unique linkage property} with respect to $Y,Z$ iff there is a set $\rset$ of $r$ node-disjoint paths in $\hat H$ connecting every vertex of $Y$ to a distinct vertex of $Z$ (that we call a $(Y,Z)$-linkage), and moreover this set of paths is unique. We say that $\hat H$ has the \emph{perfect unique linkage property} with respect to $Y,Z$ iff additionally every vertex of $\hat H$ lies on some path of $\rset$ -- the unique $(Y,Z)$-linkage in $\hat H$.
\end{definition}

Recall that graph $H'$ is the union of the paths in $\rset$ and the paths in $\qset''$. Next, we will slightly modify the graph $H'$ by contracting some of its edges, so that the resulting graph has the perfect unique linkage property with respect to $A'$ and $B'$, while preserving Property~(\ref{prop: intersection}). 
We do so by performing the following two steps:

\begin{itemize}
	\item While there is an edge $e=(u,v)$ in $H'$ that belongs to a path of $\rset$ and to a path of $\qset$, contract edge $e$ by unifying $u$ and $v$; update the corresponding paths of $\rset$ and $\qset$ accordingly.
	\item While there is a vertex $u\in V(H')$ that lies on a path of $\qset''$, but does not belong to any path in $\rset$, contract any one of the (at most two) edges incident to $v$, by unifying $v$ with one of its neighbors; update the corresponding path of $\qset''$.	
\end{itemize}

Let $H''$ be the graph obtained at the end of this procedure.

\begin{observation}\label{obs: properties of H''}
 Graph $H''$ is a minor of $H$ and Property~(\ref{prop: intersection}) holds in $H''$. 
 Moreover, $H''$ has the perfect unique linkage property with respect to $A'$ and $B'$, with the unique linkage being $\rset$.
 \end{observation}
 
 \begin{proof} 
 It is immediate to verify that $H''$ is a minor of $H'$ and hence of $H$, and that each of the two transformation steps preserve Property~(\ref{prop: intersection}). It is also immediate to verify that every vertex of $H''$ lies on some path of $\rset$, as otherwise we could execute the second step. 
 
 Finally, observe that for every edge $e\in E(H'')$ that lies on a path of $\rset$, $e$ cannot lie on any path of the original set $\qset$, since each such edge was contracted by the first step.  
Clearly, $\rset$ remains an $(A',B')$-linkage. Assume for contradiction that a different $(A',B')$-linkage $\rset'\neq \rset$ exists in $H''$. Then there is some edge $e$ that is used by the paths in $\rset$ but it is not used by the paths in $\rset'$. Therefore, $H''\setminus\set{e}$ contains an $(A',B')$-linkage of cardinality $N$, and so does $H'\setminus\set{e}$. As $e$ cannot lie on any path of the original set $\qset$, this contradicts Observation~\ref{obs: irrelevant edge}.
\end{proof}

Next, we define a new combinatorial object that will be central to this step: an $\hat M$-slicing of a set of paths.
Eventually, we will use this object to slice the paths of $\rset$. 

A convenient way to think about the $\hM$-slicing is that it will allow us to ``slice'' the pseudo-grid into $\hat M$ smaller such structures.
 
\begin{definition}
Suppose we are given a set $\hrset$ of node-disjoint paths, where for each path $R\in \hrset$, one of its endpoints $a(R)$ is designated as the first endpoint of $R$, and the other endpoint $b(R)$ is designated as the last endpoint of $R$. Given an integer  $\hat M>0$, an $\hM$-slicing $\Lambda=\set{\Lambda(R)}_{R\in\hrset}$ of $\hrset$ consists of a sequence $\Lambda(R)=(v_0(R),v_1(R),\ldots,v_{\hM}(R))$ of $\hM+1$ vertices of $R$, for every path $R\in \hrset$, such that $v_0(R)=a(R)$, $v_{\hM}(R)=b(R)$, and $v_0(R),v_1(R),\ldots,v_{\hM}(R)$ appear on $R$ in this order (we allow the same vertex to appear multiple times in $\Lambda(R)$.)
\end{definition}

Assume now that we are given some set $\hrset$ of node-disjoint paths, and another set $\hqset$ of node-disjoint paths, in some graph $\hat G$, such that each path $Q\in \hqset$ intersects at least one path of $\hrset$. Assume also that we are given an $\hM$-slicing $\Lambda=\set{\Lambda(R)}_{R\in \hrset}$ of $\hrset$.
 For all $1\leq i\leq \hM$, we denote by $\sigma_i(R)$ the sub-path of $R$ lying strictly between $v_{i-1}(R)$ and $v_i(R)$, so it excludes these two vertices (notice that it is possible that $\sigma_i(R)=\emptyset$ if $v_{i-1}(R)=v_i(R)$, or if they are consecutive vertices on $R$). For each $1\leq i\leq \hM$, let $\Sigma_i=\set{\sigma_i(R)\mid R\in \hrset}$, and let $\hqset_i\subseteq \hqset$ contain all paths $Q\in \hqset$ with the following property: for every path $R\in \hrset$, for every vertex $v\in Q\cap R$, $v\in \sigma_i(R)$. Equivalently:
 
 \[\hqset_i=\set{Q\in \hqset\mid (Q\cap \hrset)\subseteq \bigcup_{\sigma\in \Sigma_i}\sigma }.\]
 
 We say that the \emph{width} of the $\hM$-slicing $\Lambda$ with respect to $\hqset$ is $\hat w$ iff $\min_{1\leq i\leq \hM}\{|\hqset_i|\}=\hw$.
Notice that from our definition, for all $i\neq j$, $|\hqset_i\cap \hqset_j|=\emptyset$.
We now provide sufficient conditions for the existence of an $\hat M$-slicing of a given width.
 
\begin{theorem}\label{thm: slicing}
Let $\hat G$ be a graph, $\hat A,\hat B$ two sets of its vertices of cardinality $\hat N>0$ each, and assume that $\hat G$ has the perfect unique linkage property with respect to $(\hat A,\hat B)$, with the unique linkage denoted by $\hrset$. Assume that there is another set $\hqset$ of node-disjoint paths in $\hat G$, such that each path $Q\in \hqset$ intersects at least one path of $\hrset$, and integers $\hM,\hw>0$, such that $|\hqset|\geq \hM\hw+(\hM+1)\hN$. Then there is an $\hM$-slicing of $\hrset$ of width $\hw$ with respect to $\hqset$ in $\hat G$.
\end{theorem}

\begin{proof}
We use the following result of Robertson and Seymour (Lemma 2.5 from~\cite{robertson1983graph}); we note that the lemma appearing in~\cite{robertson1983graph} is somewhat weaker, but their proof immediately implies the stronger result that we state below; for completeness we include its proof in the Appendix.

\begin{lemma}
	\label{lemma: numbering}
Let $\hat G$ be a graph, $\hat A,\hat B\subseteq V(\hat G)$ two subsets of its vertices, such that $|\hat A|=|\hat B|$, and $\hat G$ has the perfect unique linkage property with respect to $(\hat A,\hat B)$, with the unique $(\hat A,\hat B)$-linkage denoted by $\hrset$. Then there is a bijection $\mu: V(\hat G)\rightarrow \set{1,\ldots,|V(\hat G)|}$ such that the following holds. For an integer $t>0$, let $S_t$ contain, for every path $R\in\hrset$, the first vertex $v$ on $R$ with $\mu(v)\geq t$; if no such vertex exists, then we add the last vertex of $R$ to $S_t$. Let $Y_t=\set{v\in V(\hat G) \mid \mu(v)<t}$ and $Z_t=\set{v \in V(\hat G)\mid\mu(v)\geq t}$. Then $\mu$ has the following properties:
\begin{itemize}
\item For each path $R\in \hrset$, for every pair $v,v'$ of its vertices, if $v'$ appears strictly before $v$ on  $R$, then  $\mu(v')<\mu(v)$; and
\item For every integer $0<t\leq |V(\hat G)|$, graph $\hat G\setminus S_t$ contains no path connecting a vertex of $Y_t$ to a vertex of $Z_t$.
\end{itemize}
\end{lemma}


We apply Lemma~\ref{lemma: numbering} to graph $\hat G$ and the sets $\hat A$ and $\hat B$ of its vertices, obtaining a bijection $\mu: V(\hat G)\rightarrow \set{1,\ldots,|V(\hat G)|}$. 


 Consider now an integer $1\leq t\leq |V(\hat G)|$, and the corresponding set $S_t$ of vertices, that we refer to as \emph{separator}. This separator contains exactly $\hN$ vertices -- one vertex from each path $R\in \hrset$.
Recall that $\hat G\setminus S_t$ contains no path connecting $Y_t=\set{v: \mu(v)<t}$ and $Z_t=\set{v:\mu(v)\geq t}$. 

We denote by $\qset^0(S_t)\subseteq \hqset$ the subset of all paths $Q\in \hqset$, such that $Q\cap S_t\neq \emptyset$, so $|\qset^0(S_t)|\leq \hN$. Let $R\in \hrset$ be some path with endpoints $a\in \hat A$, $b\in \hat B$, and let $v$ be the unique vertex of $R$ that belongs to $S_t$. Then $v$ defines two sub-paths of $R$, as follows: $R_1(S_t)$ is the sub-path of $R$ from $a$ to $v$ (including these two vertices), and $R_2(S_t)$ is similarly defined as the sub-path of $R$ from $v$ to $b$. If $S_t$ contains $b$, then $R_1(S_t)=R$ and $R_2(S_t)=(b)$. Let $\qset^1(S_t)\subseteq \hqset\setminus \qset^0(S_t)$ be the set of all paths $Q\in \hqset$, such that $Q$ intersects some path in $\set{R_1(S_t)\mid R\in \hrset}$, and let $\qset^2(S_t)$ be defined similarly for  $\set{R_2(S_t)\mid R\in \hrset}$.
Notice that equivalently, $\qset^1(S_t)$ contains all paths $Q\in \hqset$, such that $Q\cap Y_t\neq \emptyset$ and $Q\cap S_t=\emptyset$. Similarly, $\qset^2(S_t)$ contains all paths $Q\in \hqset$ with $Q\cap Z_t\neq \emptyset$ and $Q\cap S_t=\emptyset$. It is easy to verify that the paths in $\qset^1(S_t)$ are disjoint from $Z_t$, and in particular $\qset^1(S_t)\cap\qset^2(S_t)=\emptyset$: otherwise, there is some path $Q\in \hqset$, that contains a vertex of $Y_t$ and a vertex of $Z_t$, such that $Q\cap S_t=\emptyset$, contradicting the fact that $Y_t$ and $Z_t$ are separated in $\hat G\setminus S_t$.
Notice that, since every path of $\hqset$ intersects at least one path of $\hrset$, $(\qset^0(S_t),\qset^1(S_t),\qset^2(S_t))$  define a partition of $\hqset$.



The following observation will be useful in order to construct the $\hM$-slicing. 

\begin{observation}\label{obs: properties of the cuts}
The sets $\{\qset^1(S_t)\}_{t\ge 1}$ of paths satisfy the following properties:	
	
\begin{enumerate}	
\item $\qset^1(S_{1})=\emptyset$, and $\qset^1(S_{|V(\hat G)|})$ contains all but at most $|\hrset|$ paths of $\hqset$ -- the paths that intersect the vertices of $\hat B$; \label{first set is empty}

\item For all $1\leq t<t'\leq |V(\hat G)|$, $\qset^1(S_t)\subseteq \qset^1(S_{t'})$; and \label{containment of consecutive sets}

\item For all $1\leq t< |V(\hat G)|$, $|\qset^1(S_{t+1})\setminus \qset^1(S_t)|\le 1$. \label{small diffs}
\end{enumerate}
\end{observation}
\begin{proof}
From the definition of $S_1$, it  contains the first vertex from each path of $\hrset$, so $\qset^1(S_{1})=\emptyset$. 
On the other hand, $S_{|V(\hat G)|}$ contains the last vertex from each path of $\hrset$, and, since every path of $\hqset$ intersects some path of $\hrset$, $\qset^1(S_{|V(\hat G)|})$ contains all paths of $\hqset$ except those intersecting the vertices of $\hat B$. This proves the first assertion. For the second assertion, recall that set $\qset^1(S_t)$ contains all paths $Q\in \hqset$ that intersect $Y_t$ but are disjoint from $S_t$, and a path of $\qset^1(S_t)$ may not contain a vertex of $Z_t$. Each such path must intersect $Y_{t'}$, since $Y_t\subseteq Y_{t'}$, and it must be disjoint from $S_{t'}$, as $S_{t'}\subseteq S_t\cup Z_t$. Therefore, $\qset^1(S_t)\subseteq \qset^1(S_{t'})$ must hold.



We now turn to prove the third assertion. Let $Q\in \qset^1(S_{t+1})\setminus\qset^1(S_t)$.
Since $Q\not\in \qset^1(S_t)$, either (i) $Q\cap S_t\neq \emptyset$, or (ii) $Q\cap S_t=\emptyset$ and $Q\cap Y_t=\emptyset$. If the former is true, then we say that $Q$ is a type-1 path; otherwise we say that $Q$ is a type-2 path. Note that $|S_t\setminus S_{t+1}|\leq 1$, from the definition of the sets $S_{t''}$, and since $\mu$ is a bijection. Since the paths in $\hqset$ are node-disjoint, there is at most one path of type 1 that belongs to $\qset^1(S_{t+1})$. We claim that no path of type 2 may belong to $\qset^1(S_{t+1})$. 

Indeed, if $Q\in \qset^1(S_{t+1})$ is a type-2 path, then it must contain some vertex $v\in Y_{t+1}$. Assume that $v$ lies on some path $R\in \rset$. Then $\mu(v)<t+1$. If $\mu(v)<t$, then $v\in Y_t$, contradicting the fact that $Q$ is of type $2$. Therefore, $\mu(v)=t$. However, since $Q$ is of type $2$, $v\not\in S_t$, which can only happen if some vertex $v'$ that lies on the path $R$ strictly before $v$ has $\mu(v')\geq t$, and in particular $\mu(v')\geq \mu(v)$. But that is impossible from the properties of $\mu$.
%
%
%
\end{proof}

We now provide an algorithm to compute the $\hM$-slicing.
The algorithm performs $\hM-1$ iterations, where at the end of iteration $i$ we produce an integer $1\leq t_i\leq |V(\hat G)|$, and an $(i+1)$-slicing $\set{\Lambda(R)}_{R\in \hrset}$ of $\hrset$, such that the width of the slicing with respect to $\hqset$ is at least $\hw$, and the following additional properties hold:

\begin{itemize}
	\item  $|\hqset_{i+1}|\geq |\hqset|-(i+2)\hN-i\hw$; and
	\item for each path $R\in \rset$, the vertex $v_i(R)\in \Lambda(R)$ is the unique vertex of $S_{t_i}\cap R$.
\end{itemize}

Notice that the above properties ensure that $\hqset_{i+1}$ contains all paths of $\qset^2(S_{t_i})$, except for at most $\hat N$ paths, that contain the last endpoints of the paths in $\hrset$.

Since we assumed that $|\hqset|\geq \hM\hw+(\hM+1)\hN$, after $(\hM-1)$ iterations we obtain a valid $\hM$-slicing of $\hrset$ of width at least $\hw$.
 
In order to execute the first iteration, we 
let $t_1>0$ be an integer, for which $|\qset^1(S_{t_1})|= \hw+\hN$. Such an integer must exist from Observation~\ref*{obs: properties of the cuts}.  For all $R\in \hrset$, we let $v_1(R)$ be the unique vertex of $R$ lying in $S(t_1)$, and $v_0(R), v_2(R)$ the endpoints of $R$ lying in $\hat A$ and $\hat B$, respectively. This immediately defines a $2$-slicing of the paths in $\hrset$. Recall that for each path $R\in \hrset$, we obtain two segments: $\sigma_1(R)$, that is obtained from $R_1(S_{t_1})$ by removing its two endpoints, and $\sigma_2(R)$, obtained similarly from $R_2(S_{t_1})$. It is immediate to verify that set $\hqset_1$ of paths associated with this slicing contains every path $Q\in\qset^1(S_{t_1})$, except for those paths that contain the first endpoints of the paths in $\hrset$. Therefore, $|\hqset_1|\geq |\qset^1(S_{t_1})|-\hN\geq \hw$. Similarly, set $\hqset_2$ of paths associated with this slicing contains every path $Q\in \qset^2(S_{t_1})$, except for those paths that contain the last endpoints of the paths in $\hrset$. Therefore, $|\hqset_2|\geq |\hqset|-|\qset^0(S_{t_1})|-|\qset^1(S_{t_1})|-|\hrset|\geq |\hqset|-3\hN-\hw$, as required.

We now fix some $1\le i<\hM-1$, and describe the $(i+1)$th iteration. We assume that we are given an $(i+1)$-slicing  of $\hrset$, with $\set{v_{i}(R)\mid R\in \hrset}=S_{t_i}$, and $|\hqset_{i+1}|\geq |\hqset|-(i+2)\hN-i\hw\geq \hN+2\hw$.
 
Let $t_{i+1}$ be the integer, for which $ |\qset^1(S_{t_{i+1}})\cap \hqset_{i+1}|=\hw$.
Such an integer must exist from Observation~\ref{obs: properties of the cuts}, since $|\hqset_{i+1}|\geq 2\hw+\hN$. Moreover, we are guaranteed that $t_{i+1}>t_i$, since $\hqset_{i+1}\subseteq \qset^2(S_{t_i})$.

For convenience, we denote $\qset^0(S_{t_i}),\qset^1(S_{t_i})$ and $\qset^2(S_{t_i})$ by $\qset^0,\qset^1$ and $\qset^2$, respectively.

For every path $R\in \hrset$, vertices $v_0(R),\ldots,v_i(R)$ remain the same as before. We let $v_{i+1}(R)$ be the unique vertex of $R\cap S_{t+1}$, and we let $v_{i+2}(R)$ be the endpoint of $R$ lying in $\hat B$. We now obtain a new $(i+2)$-slicing $\Lambda'=\set{\Lambda'(R)}_{R\in \hrset}$, 
where for each path $R\in \hrset$, $\Lambda'(R)=(v_0(R),\ldots,v_{i+1}(R))$. For convenience, let $\hqset'_{i+1}$ denote the original set $\hqset_{i+1}$, and let $\hqset_{i+1}$ denote the new set, defined with respect to the new slicing. Then $\hqset_{i+1}$ contains all paths of $\qset^1\cap \hqset'_{i+1}$, and so, from the definition of $t_{i+1}$, $|\hqset_{i+1}|= \hw$. Set $\hqset_{i+2}$ contains all paths of $\hqset'_{i+1}\setminus\hqset_{i+1}$, except for the paths containing the vertices of $\qset^0$ -- there are at most $\hN$ such paths. Therefore, $|\hqset_{i+2}|\geq |\hqset'_{i+1}|-|\hqset_{i+1}|-\hN\geq |\hqset|-(i+2)\hN-i\hw-\hw-\hN\geq |\hqset|-(i+3)\hN-(i+1)\hw$, as required.
 \end{proof}

 Let $M=8g^4\log g$. From Theorem~\ref{thm: slicing}, we can obtain an $M$-slicing $\Lambda=\set{\Lambda(R)}_{R\in \rset}$ of $\rset$, of width $w=2^{11}g^6$ with respect to $\qset''$, since $N\leq g^2 D=64g^6$, and so $|\qset''|\geq \kappa/8\geq 2^{19}g^{10}\log g\geq M(2N+w)$. For every $1\leq i\leq M$, we denote by $\Sigma_i=\set{\sigma_i(R)\mid R\in \rset}$.
 We denote the subset $\hqset_i\subseteq \qset''$ of paths corresponding to $\Sigma_i$ by $\qset_i$. We call $(\Sigma_i,\qset_i)$ the \emph{$i$th slice of $\Lambda$}.

 
 

\subsection{Step 3: Intersecting Path Sets}
We start by defining $(\hw,\hD)$-intersecting pairs of path sets.

\begin{definition}
	Let $\hrset$, $\hqset$ be two sets of node-disjoint paths in a graph $\hat G$. Given integers $\hw,\hD>0$, we say that $(\hrset,\hqset)$ is a $(\hw,\hD)$-intersecting pair of path sets iff each path $R\in \hrset$ intersects at least $\hw$ distinct paths of $\hqset$, and each path $Q\in \hqset$ intersects at least $\hD$ distinct paths of $\hrset$.
\end{definition}

\begin{lemma}\label{lemma: intersecting sets}
	Let $\hrset$, $\hqset$ be two sets of node-disjoint paths in a graph $\hat G$, and let $\hw,\hD>0$ be integers. Assume that each path $Q\in \hqset$ intersects at least $2\hD$ distinct paths of $\hrset$, and that $|\hqset|\geq 2|\hrset|\hw/\hD$. Then there is a partition $(\hrset',\hrset'')$ of  $\hrset$, and a subset $\hqset'\subseteq \hqset$ of paths, such that $(\hrset',\hqset')$ is a $(\hw,\hD)$-intersecting pair of path sets; $|\hqset'|\geq |\hqset|/2$; and every path in $\hrset''$ intersects at most $\hw$ paths of $\hqset'$.
\end{lemma}

\begin{proof}
	We start with $\hrset'=\hrset$ and $\hqset'=\hqset$, and then iterate, by performing one of the following two operations as long as possible:
	
	\begin{itemize}
		\item If there is a path $R\in \hrset'$ intersecting fewer than $\hw$ distinct paths of $\hqset'$, delete $R$ from $\hrset'$.
		
		\item If there is a path $Q\in \hqset'$ intersecting fewer than $\hD$ distinct paths of $\hrset'$, delete $Q$ from $\hqset'$.
		 
	\end{itemize}

Clearly, when the algorithm terminates, $(\hrset',\hqset')$ are a $(\hw,\hD)$-intersecting pair of path sets, and each path in $\hrset''=\hrset\setminus\hrset'$ intersects at most $\hw$ paths of $\hqset'$. It now remains to prove that $|\hqset'|\geq |\hqset|/2$.

Let $\Pi\subseteq \hrset\times\hqset$ be the set of all pairs $(R,Q)$ of paths, such that $R\cap Q\neq\emptyset$. We call each pair $(R,Q)\in \Pi$ an \emph{intersection}. When a path $R$ is deleted from $\hrset'$, it participates in at most $\hw$ intersections. We say that $R$ is \emph{responsible} for these intersections, and that these intersections are deleted due to $R$. Overall, all paths of $\hrset\setminus\hrset'$ may be responsible for at most $|\hrset|\hw$ intersections.

Consider now some path $Q\in \hqset\setminus\hqset'$. Originally, $Q$ intersected at least $2\hD$ paths of $\hrset$, but at the time it was removed from $\hqset'$ it intersected at most $\hD$ such paths. Therefore, at least $\hD$ of its intersections were removed, and these intersections must have been removed due to paths in $\hrset\setminus\hrset'$. Therefore, $|\hqset\setminus\hqset'|\leq |\hrset|\hw/\hD\leq \hqset/2$.
\end{proof}

For each $1\leq i\leq M$, We apply Lemma~\ref{lemma: intersecting sets} to sets $(\Sigma_i,\qset_i)$ of paths, with parameters $\hw=4g^2$ and $\hD=D/2$. Notice that $|\Sigma_i|\leq |\rset|= N\leq Dg^2$, while $|\qset_i|\geq 2^{11}g^6$. It is then easy to verify that $|\qset_i|\geq 2|\Sigma_i|\hat w /\hat D=16|\Sigma_i|g^2/D$. Recall that each path $Q\in\qset_i$ intersects at least $D$ paths of $\Sigma_i$. Therefore, we obtain a partition $(\Sigma'_i,\Sigma''_i)$ of $\Sigma_i$, and a subset, $\qset'_i\subseteq\qset_i$ of paths, such that $(\Sigma'_i,\qset'_i)$ is a $(4g^2,D/2)$-intersecting pair of path sets, $|\qset'_i|\geq 2^{10}g^6$, and each path of $\Sigma_i''$ intersects at most $4g^2$ paths of $\qset_i'$. For convenience, we provide the list of the main parameters of the current section in Subsection~\ref{subsec: params weaker}.

From now on, the remainder of the proof will only use the following facts. We are given a set $\rset$ of $N$ disjoint paths, where $D\leq N\leq Dg^2$ and $D=64g^4$. We are also given a slicing $\Lambda$ of $\rset$ into $M=8g^4\log g$ slices. For each $1\leq i\leq M$, set $\Sigma_i$ contains the $i$th segment $\sigma_i(R)$ of every path $R\in \rset$. We are also given a subset $\Sigma'_i\subseteq \Sigma_i$, and another set $\qset_i'$ of node-disjoint paths, such that $(\Sigma'_i,\qset_i')$ are $(4g^2,D/2)$-interesecting. The paths in $\qset_i'$ are disjoint from all paths in $\bigcup_{j\neq i}(\qset_j'\cup \Sigma_j)$. The remainder of the proof only uses these facts, and in particular it does not depend on the initial value of the parameter $\kappa$ or on the cardinalities of sets $\qset'_i$ of paths. We will use this fact in Section~\ref{sec: stronger structural theorem} when we prove the stronger version of Theorem~\ref{thm: main for building crossbar}.

 \subsection{Step 4: Well-Linked Decomposition}\label{subsec: WLD}
 In this step, we need to use a slightly weakened definition of edge-well-linkedness, that was also used in previous work.
 
 \begin{definition}
 	Let $\hG$ be a graph, $T$ a subset of its vertices, and $\hw>0$ an integer. We say that $T$ is $\hw$-weakly well-linked in  $\hG$, iff for any two disjoint subsets $T',T''$ of $T$, there is a set of $\min\set{|T'|,|T''|,\hat w}$ edge-disjoint paths in $\hG$, connecting vertices of $T'$ to vertices of $T''$.
 \end{definition}
 
 The following observation is immediate.
 
 \begin{observation}\label{obs: weak to edge wl}
 	Let $\hG$ be a graph, $T$ a subset of its vertices, and $\hw>0$ an integer, such that $|T|\leq 2\hw$. Assume further that $T$ is $\hw$-weakly-well-linked in $\hG$. Then $T$ is edge-well-linked in $\hG$.
 \end{observation}

 We will repeatedly use the following simple observation.
 \begin{observation}\label{obs: cut from wl}
 	Let $\hG$ be a graph,  $T$ a subset of its vertices, and $\hw>0$ an integer. Assume that $T$ is {\bf not} $\hat w$-weakly well-linked in $\hG$. Then there is a partition $(X,Y)$ of $V(\hG)$, such that $|E(X,Y)|< \min\set{\hw, |T\cap X|,|T\cap Y|}$.
 \end{observation}

\begin{proof}
	Since $T$ is not $\hw$-weakly well-linked in $\hG$, there are two disjoint subsets $T',T''$ of $T$, such that the largest number of edge-disjoint paths connecting vertices of $T'$ to vertices of $T''$ contains fewer than $b=\min\set{\hw, |T'|,|T''|}$ paths. From Menger's Theorem, there is a set $E'$ of at most $b-1$ edges, such that $\hG\setminus E'$ contains no path connecting a vertex of $T'$ to a vertex of $T''$. Therefore, there is a partition $(X,Y)$ of $\hG$, with $E(X,Y)\subseteq E'$, $T'\subseteq X$, and $T''\subseteq Y$. It is easy to verify that this partition has the required properties.
\end{proof}

 Let $\hat G$ be a graph, and let $\hat \Sigma$ be a set of node-disjoint paths in $\hat G$. 
 Given a sub-graph $C\subseteq \hat G$, we denote by $\hat\Sigma(C)$ the set of all paths $\sigma\in \hat \Sigma$, such that $\sigma\subseteq C$, and we denote by $\Gamma(C)$ the set of endpoints of all paths in $\hat \Sigma(C)$. We sometimes refer to sub-graphs $C\subseteq \hat G$ as \emph{clusters}.
 
 Given two parameters, $\hat w$ and $\hD$, we say that cluster $C$ is \emph{good} iff $\Gamma(C)$ is $\hat w$-weakly well-linked in $C$. We say that it is \emph{happy}, if it is good, and additionally, $|\hat \Sigma(C)|\geq \hat D$.
 Following is the main theorem of this step.
 
 \begin{theorem}\label{thm: wld}
 Let $\hat G$ be a graph, and let $\hat \Sigma$ and $\hat \qset$ be two sets of node-disjoint paths in $\hat G$ (but a path in $\hat \Sigma$ and a path in $\hat \qset$ may intersect). Let $\hat w,\hat D>0$ be integers, such that $(\hat \Sigma,\hat \qset)$ are $(4\hat w,2\hat D)$-intersecting, and $\hD\geq 8\hw$. Then there is a collection $\cset$ of disjoint sub-graphs of $\hat G$, and a subset $\hat\Sigma'\subseteq \hat \Sigma$, such that:
 
 \begin{itemize}
 \item Each cluster $C\in \cset$ is happy (that is, $|\hat\Sigma(C)|\geq \hat D$ and $\Gamma(C)$ is $\hw$-weakly well-linked in $C$);
 \item $|\hat \Sigma'|\geq |\hat \Sigma|/4$; and
 \item every path $\sigma\in \hat \Sigma'$ belongs to a set $\hat \Sigma(C)$ for some $C\in \cset$.
 \end{itemize}
 \end{theorem}
 
 \begin{proof}
Throughout the algorithm, we maintain a set $\cset$ of disjoint clusters of $\hat G$. Recall that $\hat \Sigma(C)$ contains all paths $\sigma\in \hat \Sigma$, such that $\sigma\subseteq C$. At the beginning, $\cset$ contains a single cluster $\hG$, and $\hat \Sigma(\hG)=\hat \Sigma$. We also maintain a set $E'$ of edges that we have deleted, that is initialized to $\emptyset$. The algorithm is executed as long as there is some cluster $C\in \cset$, such that $\Gamma(C)$ is not $\hw$-weakly well-linked in $C$.

Let $C\in \cset$ be any such cluster. For convenience, denote $T=\Gamma(C)$. From Observation~\ref{obs: cut from wl}, there is a partition $(X,Y)$ of $V(C)$, such that $|E(X,Y)|< \min\set{\hw,|T\cap X|,|T\cap Y|}$. Notice that in this case, each of $\hG[X]$ and $\hG[Y]$ must contain at least one path of $\hat \Sigma(C)$. Indeed, assume for contradiction that no path of $\hat \Sigma(C)$ is contained in $\hG[X]$. Then for every vertex $v\in T\cap X$, path $\sigma\in \hat \Sigma(C)$ that contains $v$ as an endpoint must contribute at least one edge to $E'$. Moreover, if both endpoints of $\sigma$ belong to $X$, then at least two edges of $\sigma$ lie in $E'$. Therefore, $|E'|\geq |T\cap X|$, a contradiction.

 We add the edges of $E(X,Y)$ to $E'$. Let $J\subseteq \hat \Sigma(C)$ be the set of all paths that contain edges of $E(X,Y)$. Each path of $\hat\Sigma(C)\setminus J$ is now either contained in $\hG[X]$ or in $\hG[Y]$. We remove $C$ from $\cset$ and replace it with $\hG[X]$ and $\hG[Y]$. This finishes the description of an iteration.
Let $\cset$ be the final set of clusters at the end of the algorithm, and let $|\cset|=r$. Then our algorithm has executed $r-1$ iterations. Observe that in each iteration at most $\hw$ edges are added to $E'$, so at the end of the algorithm, $|E'|\leq (r-1)\hw$. For every cluster $C\in \cset$, let $\out(C)$ be the set of all edges of $E'$ that are incident to $C$.

We partition all clusters in $\cset$ into two subsets: $\cset_1\subseteq\cset$ contains all clusters with $|\out(C)|< 4\hw$, and $\cset_2$ contains all remaining clusters of $\cset$.

\begin{observation}\label{obs: number of happy clusters}
$|\cset_1|\geq r/2$.
\end{observation}

\begin{proof}
Assume otherwise. Then $|\cset_2|>r/2$, and so $|E'|=\sum_{C\in \cset}|\out(C)|/2>(r/2)(4\hw)/2=r\hw$. However, as observed above, $|E'|\leq (r-1)\hw$, a contradiction.
\end{proof}

\begin{observation}\label{obs: happy clusters}
Every cluster $C\in \cset_1$ is happy.
\end{observation}
\begin{proof}
Consider some cluster $C\in \cset_1$. Since the algorithm has terminated, $\Gamma(C)$ is $\hw$-weakly-well-linked in $C$. It is now enough to show that $|\hat\Sigma(C)|\geq \hD$. Recall that our procedure guarantees that $\hat\Sigma(C)\neq \emptyset$. Let $\sigma\in \hat\Sigma(C)$ be any path. Since $(\hat \Sigma,\hqset)$ are $(4\hw,2\hD)$-intersecting, there are at least $4\hw$ paths in $\hqset$ that intersect $\sigma$. Since $|\out(C)|<4\hw$, and the paths in $\hqset$ are disjoint, at least one such path $Q$ is contained in $C$. Path $Q$, in turn, intersects at least $2\hD$ paths of $\hat\Sigma$. Since $|\out(C)|<4\hw$, all but at most $4\hw$ these paths are contained in $C$. Therefore, $|\hat\Sigma(C)|\geq 2\hD-4\hw\geq \hD$, since we have assumed that $\hD\geq 8\hw$.
\end{proof}

\begin{corollary}
$r\leq 2|\hat\Sigma|/\hD$.
\end{corollary}

\begin{proof}
Note that from Observations~\ref{obs: number of happy clusters} and~\ref{obs: happy clusters}, $\sum_{C\in \cset_1}|\hat\Sigma(C)|\geq r\hD/2$. On the other hand, $\sum_{C\in \cset_1}|\hat\Sigma(C)|\leq |\hat \Sigma|$. The corollary now follows.
\end{proof}

We say that a path $\sigma\in \hat \Sigma$ is destroyed if at least one of its edges belongs to $E'$; otherwise, we say that it survives. From the above corollary, the number of paths that are destroyed is bounded by $r\hw\leq 2\hw|\hat \Sigma|/\hD\leq |\hat\Sigma|/2$, since we have assumed that $\hD\geq 8\hw$. Each of the surviving paths belongs to some set $\hat\Sigma(C)$ for $C\in \cset$. At least half the clusters in $\cset$ are happy. For a happy cluster $C$, $|\hat\Sigma(C)|\geq \hD$, and for an unhappy cluster $C$, $|\hat\Sigma(C)|<\hD$. We denote by $\hat\Sigma'\subseteq \hat \Sigma$ the set of all paths $\sigma$, such that $\sigma\in \hat\Sigma(C)$ for a happy cluster $C$. Then $\hat\Sigma'$ contains at least half of the surviving paths, and altogether, $|\hat \Sigma'|\geq |\hat \Sigma|/4$.
\end{proof}

Recall that we have $M=8g^4\log g$ slices $\{(\Sigma_i,\qset_i)\}_{1\le i\le M}$. Recall that for each $1\leq i\leq M$, we have computed subsets $\Sigma'_i\subseteq\Sigma_i$ and $\qset'_i\subseteq \qset_i$, such that $(\Sigma'_i,\qset'_i)$ are $(4g^2,D/2)$-intersecting. Let $H'_i$ be a graph obtained from the union of the paths in $\Sigma_i$ and $\qset_i$. 
Denote $\hat D=D/4$ and $\hat w=g^2$, so that $(\Sigma'_i,\qset'_i)$ are $(4\hw,2\hD)$-intersecting. Note that $\hat D\geq 8\hat w$, since $D=64g^4$.
We apply Theorem~\ref{thm: wld} to graph $H'_i$, with $\hat \Sigma=\Sigma'_i$, $\hat \qset=\qset'_i$, and parameters $\hD$, $\hw$. Let $\cset_i$ be the resulting collection of happy clusters, and let $C_i\in \cset_i$ be any such cluster. We denote by $\tilde \Sigma_i=\hat\Sigma(C_i)$. Recall that $|\tilde \Sigma_i|\geq D/4$, and the endpoints of the paths in $\tilde \Sigma_i$ are $g^2$-weakly well-linked in $C_i$. Finally, we let $\tilde \cset=\set{C_1,\ldots,C_{M}}$. To summarize, we have obtained a collection of $M$ clusters, one cluster per slice. Each cluster $C_i$ contains a set $\tilde \Sigma_i$ of at least $D/4$ segments, whose endpoints are $g^2$-weakly well-linked in $C_i$.

\subsection{Step 5: Constructing a Weak Path-of-Sets System}
\label{subsec:weakPoS}

In this section we construct a weak \PoS of width $g^2$ and length $g^2$ in graph $H''$. Recall that $H''$ is a minor of the graph $H$, and it is the union of the paths in $\rset$ and $\qset''$. In particular, the maximum vertex degree in $H''$ is at most $4$. We will use this fact in the last step, in order to turn the \PoS into a strong one. For an integer $\hN>0$, we denote $[\hN]=\set{1,\ldots,\hN}$. Let $\pi: \rset\rightarrow [N]$ be an arbitrary bijection, mapping every path in $\rset$ to a distinct integer of $[N]$ (recall that $|\rset|=N$). 
Following is the main theorem for this step.

\begin{theorem} \label{thm: long intersecting chain}
Let $\hD,\hN,\hM,\hw$ be non-negative integers, such that (i) $\hN\geq 3\hw$; (ii) $\hD^2\ge 4\hat{N} \hw$; and (iii) $\hM\hD\ge 2\hat{N}\hw$ hold, 
and let $S_1,\ldots,S_{\hM}$ be subsets of $[\hN]$, where for each $1\leq i\leq \hM$, $|S_i|\geq \hD$. Then there are $\hw$ indices $1\leq i_1< i_2<\cdots<i_{\hw}\leq \hM$, such that for all $1\leq j<{\hw}$, $|S_{i_j}\cap S_{i_{j+1}}|\ge \hw$.
\end{theorem}

We prove this theorem below, after we show how to construct a weak \PoS of length and width $g^2$ in $H''$ using it.
Denote $\hat M=M=8g^4\log g$, $\hat D=D/4$, $\hat N=N$, and $\hw=g^2$. Recall that we are given a set $\cset=\set{C_1,C_2,\ldots,C_{\hM}}$ of clusters, where cluster $C_i$ corresponds to slice $i$. Recall also that for each $i$, $|\tilde\Sigma(C_i)|\geq D/4$. We now build the subsets $S_1,\ldots,S_{\hM}$ of $[\hat N]$ as follows. Fix some $1\leq i\leq \hM$. For every path $R\in \rset$, such that $\sigma_i(R)\in\tilde{\Sigma}(C_i)$, we add $\pi(R)$ to $S_i$. Notice that for all $i$, $|S_i|\geq D/4=\hD$. 
We now verify that the conditions of Theorem~\ref{thm: long intersecting chain} hold for the chosen parameters. The first condition, $\hN\geq 3\hw$, is immediate from the fact that $N\geq D=64g^4$. The second condition,
 $\hD^2\geq 4\hat{N}\hw$, is equivalent to: $D^2/16\geq 4Ng^2$. Since $N\leq Dg^2$, it is enough to show that $D\geq 64g^4$, which holds from the definition of $D$. The third condition, $\hD\hM\geq 2\hat{N}\hw$ is equivalent to: $DM\geq 8Ng^2$. Using the fact that $N\leq Dg^2$, and that $M\geq 8g^4$, the inequality clearly holds.

Therefore, we can now apply Theorem~\ref{thm: long intersecting chain} to conclude that there are indices $1\leq i_1< i_2<\cdots<i_{g^2}\le \hM$, such that for all $1\leq j<g^2$, $|S_{i_j}\cap S_{i_{j+1}}|\ge g^2$.

Next, we define, for all $1\leq j\leq g^2$, subsets $T_j\subseteq S_{i_j}$ of $g^2$ indices, as follows. Set $T_1$ is an arbitrary subset of $g^2$ indices of $S_{i_1}$. For each $1<j\leq g^2$, set $T_j$ is an arbitrary subset of $g^2$ indices in $S_{i_{j-1}}\cap S_{i_{j}}$.  For convenience, we also define a set $T_{g^2+1}=T_{g^2}$.
 
 The clusters $C'_1,\ldots,C'_{g^2}$ of the Path-of-Sets system are defined as follows. For $1\leq j\leq g^2$, cluster $C'_j=C_{i_j}$. 
Observe that for all $1\leq j\leq g^2+1$, set $T_j$ of indices defines a subset $\trset^j\subseteq \rset$ of paths: these are all paths $R\in \rset$ with $\pi(R)\in T_j$. Therefore, $|\trset^j|=g^2$, and for each path $R\in \trset^j$, its segment $\sigma_{i_j}(R)\in\tilde \Sigma(C_{i_j})$. Moreover, if $j>1$, then for every path $R\in \trset^j$,  $\sigma_{i_{j-1}}(R)\in \tilde\Sigma(C_{i_{j-1}})$.

 Consider now some index $1\leq j\leq g^2$. Let $\Sigma^A_{i_j}\subseteq \tilde \Sigma_{i_j}$ denote all segments $\sigma_{i_j}(R)$ of paths $R\in \trset^j$, and let $A_j$ be the set of vertices containing the first endpoint of each such segment. Let $\Sigma^B_{i_j}\subseteq\tilde \Sigma_{i_j}$ denote all segments $\sigma_{i_j}(R)$ of paths $R\in \trset^{j+1}$, and let $B_j$ be the set of vertices containing the last endpoint of each such segment. Then from Theorem~\ref{thm: wld}, the endpoints of the paths in $\Sigma^A_{i_j}\cup \Sigma^B_{i_j}$ are $g^2$-weakly well-linked in $C'_j=C_{i_j}$. Therefore,  $A_j\cup B_j$ is $g^2$-weakly well-linked in $C'_j$. Moreover, since $|A_j|=|B_j|=g^2$, from Observation~\ref{obs: weak to edge wl}, $A_j\cup B_j$ is edge-well-linked in $C'_j$.

It now remains to construct, for each $1\leq j<g^2$, the set $\pset_j$ of disjoint paths, connecting every vertex of $B_j$ to a distinct vertex of $A_{j+1}$. Recall that $|B_j|=|A_{j+1}|=|\trset^{j+1}|$, and every path $R\in \trset^{j+1}$ contains a single vertex $b_R$ of $B_j$ and a single vertex $a_R$ of $A_{j+1}$. For each path $R\in \trset^{j+1}$, we add the sub-paths of $R$ between $b_R$ and $a_R$ to $\pset_j$. It is easy to verify that all paths in set $\bigcup_j\pset_j$ are disjoint from each other and are internally disjoint from the clusters $C'_j$. This is because for each $1\leq j< g^2$, for each path $P\in \pset_j$, $P$ is a sub-path of some path $R\in \trset^{j+1}$ spanning its segments $\sigma_{i_{j}+1}(R),\ldots,\sigma_{i_{j+1}-1}(R)$, and two additional edges, one immediately preceding $\sigma_{i_{j}+1}(R)$, and one immediately following $\sigma_{i_{j+1}-1}(R)$. Since $i_1<i_2<\cdots<i_{g^2}$, the paths in $\bigcup_j\pset_j$ are disjoint from each other, and are internally disjoint from $\bigcup_{j'}C'_{j'}$.

In order to complete the proof of Theorem~\ref{thm: main for building crossbar}, it is now enough to prove Theorem~\ref{thm: long intersecting chain}.



\begin{proofof}{Theorem~\ref{thm: long intersecting chain}}	
The proof consists of three steps. First, we use the sets $S_1,\ldots,S_{\hM}$ of indices to define a directed acyclic graph. Then we show that the size of the maximum independent set in this graph is small. We use this fact to conclude that the graph must contain a long directed path, which is then used to construct the desires sequence $i_1,\ldots,i_w$ of indices.

We start by defining a directed graph $\hat{G}=(\hat{V}, \hat{E})$, where $\hat{V}=\{1,2,\cdots, \hM\}$, and for every pair $1\leq i<j\leq \hM$ of its vertices, we add a directed edge $(i,j)$ to $\hat E$ if and only if $|S_i\cap S_j|\ge \hw$. It is easy to verify that $\hat{G}$ is indeed a directed acyclic graph. We  will use the following claim to bound the size of the maximum independent set in $\hat G$.

\begin{claim}\label{claim: independent set}
Let $\sset\subseteq \set{S_1,\ldots,S_{\hM}}$ be any collection of $r=\ceil{2\hN/\hD}$ sets. Then there are two distinct sets $S_i,S_j\in \sset$ with $|S_i\cap S_j|\geq \hw$.
\end{claim}

The claim immediately implies the following corollary.

\begin{corollary}\label{cor: independent set}
Let $V'\subseteq \hat V$ be any subset of vertices of $\hat G$, such that no two vertices in $V'$ are connected by an edge. Then $|V'|< 2\hN/\hD$.
\end{corollary}

We now turn to prove Claim~\ref{claim: independent set}.

\begin{proofof}{Claim~\ref{claim: independent set}}
Assume for contradiction that the claim is false. Then there must exist $r$ subsets $U_1,\ldots,U_r$ of $[\hN]$, each of which has cardinality at least $\hD$, such that for every pair $U_i,U_j$ of the subsets, with $i\neq j$, $|U_i\cap U_j|<\hw$. We can assume without loss of generality that for each $1\leq i\leq r$, $|U_i|=\hD$. Notice that, if $\hD\geq 2\hN/3$, then every pair $U_i,U_j$ of sets must share at least $\hN/3\geq \hw$ elements, leading to a contradiction. Therefore, we assume from now on that $\hD<2\hN/3$.
Let $s=\sum_{i=1}^r|U_i|=r\hD$. 

For every element $j\in [\hat{N}]$, let $n_j$ denote the total number of sets $U_i$ with $j\in U_i$, so $\sum_{j\in [\hat{N}]}n_j=\sum_{i=1}^r|U_i|=s$. 
Let $\Pi$ be the set of all triples $(i,i',j)$, with $1\leq i, i'\leq r$ and $j\in [\hat{N}]$, such that element $j$ belongs to both $U_i$ and $U_{i'}$. Notice that each element $j$ contributes exactly $n_j^2$ triples of the form $(i,i',j)$ to $\Pi$, and so $|\Pi|=\sum_{j=1}^{\hN}n_j^2$. Since the function $f(x)=x^2$ is convex, $\frac {1}{\hat{N}}\sum_{j=1}^{\hN}n_j^2\geq \left(\sum_{j=1}^{\hN}\frac{n_j}{\hat{N}}\right )^2=\left (\frac{s}{\hat{N}}\right )^2$. Therefore:

\begin{equation}
|\Pi|\geq \frac{s^2}{\hat{N}}. \label{eq: upper bound on Pi}
\end{equation}

On the other hand, consider any pair $1\leq i,i'\leq r$ of indices. If $i=i'$, then $U_i$ contributes $|U_i|=\hD$ triples of the form $(i,i,j)$ to $\Pi$, and the total number of triples of the form $(i,i,j)$ for all $1\leq i\leq r$ and $j\in [\hat{N}]$ is $r\hD=s$. Otherwise,  $i\neq i'$, and $U_i$ and $U_{i'}$ contribute fewer than $2\hw$ triples of the form $(i,i',j)$ or $(i',i,j)$ to $\Pi$. Therefore, 

\begin{equation}
\begin{split}
|\Pi|&<  s+{r\choose 2}\cdot 2\hw \\& \leq s+r^2\hw\\
&= s +\hw\cdot\left(\ceil{\frac{2\hN}{\hD}}\right) ^2\\
&\leq s+\hw\cdot \frac{8\hN^2}{\hD^2}.
\end{split} \end{equation}

(we have used the fact that $\hD<2\hN/3$ for the last inequality).
Altogether, we get that:

\begin{equation} \frac{s^2}{\hat{N}} \leq |\Pi|< s+\frac{8\hw\hN^2}{\hD^2}.\label{eq: full}\end{equation}

We claim that:

\begin{equation}
\frac{s^2}{\hat{N}}\geq 2s\label{eq: bound 1}
\end{equation}

and that:

\begin{equation}
\frac{s^2}{\hat{N}}\geq \frac{16\hw\hN^2}{\hD^2}\label{eq: bound 2}
\end{equation}

Notice that, if both inequalities are true, then their average contradicts Inequality~(\ref{eq: full}). 
We now show that both inequalities are indeed true. The first inequality is equivalent to $s\geq 2N$. Since $s=r\hD=\ceil{\frac{2\hN}{\hD}}\cdot\hD\geq 2\hN$, this inequality is clearly true.

For the second inequality, since $s\geq 2\hN$, we get that $\frac{s^2}{\hN}\geq 4\hN$. Since we have assumed that $\hD^2\geq 4\hN\hw$, we get that $\frac{s^2}{\hN}\geq 4\hN\geq \frac{16\hw\hN^2}{\hD^2}$.

\end{proofof}

Next, we show that graph $\hat G$ contains a long directed path. We say that a subset $V'$ of vertices of a directed graph is an \emph{independent set} iff no pair of vertices in $V'$ is connected by an edge.
	
\begin{claim}
\label{claim:ind-len}
Let  $G=(V,E)$ be any directed acyclic graph on $\hM$ vertices. Let $\ell(G)$ be the length of the longest directed path in $G$, and let $\alpha(G)$ be the cardinality of the largest independent set in $G$. Then $\ell(G)\geq \hM/\alpha(G)$.
\end{claim}

\begin{proof}
We construct a partition of the vertices of $G$ into subsets $L_1,L_2,\ldots,L_h$ for some integer $h$, such that:

\begin{itemize}
\item For all $1\leq i\leq h$, set $L_i$ is non-empty and it is an independent set; and

\item For all $i>1$, for every vertex $v\in L_i$, there is some vertex $v'\in L_{i-1}$, such that the edge $(v',v)$ belongs to $G$.
\end{itemize}

Notice that for all $i$, $|L_i|\leq \alpha(G)$, and, since every vertex of $G$ belongs to some subset $L_i$, we get that $h\geq \hM/\alpha(G)$. It is now immediate to construct a directed path in $G$ of length $h$, by starting at any vertex $v\in L_h$, and then iteratively moving to a vertex in the previous subset that connects to the current vertex with an edge.

We now describe the construction of the sets $L_i$. Set $L_1$ contains all vertices of $G$ with no incoming edges. It is easy to verify that it is a non-empty independent set. Assume now that we have defined the sets $L_1,\ldots,L_{i-1}$. If every vertex of $G$ belongs to one of these sets, then we terminate the algorithm. Otherwise, consider the graph $G_i=G\setminus(L_1\cup\cdots\cup L_{i-1})$. We then let $L_i$ contain all vertices of $G_i$ that have no incoming edges in $G_i$. Notice that $L_i$ must be an independent set in the original graph $G$. Moreover, for every vertex $v\in L_i$, there must be some vertex $v'\in L_{i-1}$ with $(v',v)\in E(G)$, as otherwise $v$ should have been added to $L_{i-1}$.
\end{proof}

We conclude that graph $\hat G$ has a directed path of length at least $\hM\hD/2\hN\geq \hw$ (we have used the assumption that $\hM\hD\geq 2\hN\hw$ from the statement of Theorem~\ref{thm: long intersecting chain}). This directed path immediately defines the desired sequence $i_1,\ldots,i_{\hat w}$ of indices, completing the proof of Theorem~\ref{thm: long intersecting chain}.
\end{proofof}

\subsection{Step 6: a Strong Path-of-Sets System}
\label{subsec: strongPoS}

Recall that in Step 5, we have constructed a weak \PoS of length $g^2$ and width $g^2$, that we denote by $\pos=(\cset,\set{\pset_i}_{i=1}^{g^2-1},A_1,B_{g^2})$, in a minor $H''$ of $H$, whose maximum vertex degree is bounded by $4$. 
	Let $\ell=w=g^2$.
	Abusing the notation, we denote $\cset=(C_1,\ldots,C_{\ell})$.
	In this step we complete the proof of (the weaker version of) Theorem~\ref{thm: main for building crossbar}, by converting $\pos$ into a strong \PoS. The length of the new \PoS will remain $\ell$, and the set $\cset$ of clusters will remain the same. The width will decrease by a constant factor.
	
	This step uses standard techniques, and is mostly identical to similar steps in previous proofs of the Excluded Grid Theorem of~\cite{CC14,GMT-STOC,chuzhoy2016improved}. In particular, we will use the Boosting Theorems of~\cite{CC14} in order to select large subsets $\pset'_i\subseteq \pset_i$ of paths, such that their endpoints are sufficiently well-linked in their corresponding clusters.
	Chekuri and Chuzhoy~\cite{CC14} employed the following definition of well-linkedness:

	\begin{definition}
		We say that a set $T$ of vertices is
		$\alpha$-well-linked in a graph $G$, if for every
		partition $(A,B)$ of the vertices of $G$ into two subsets,
		$|E(A,B)|\geq \alpha\cdot \min\set{|A\cap T|,|B\cap T|}$.
	\end{definition}

The next observation follows immediately from Menger's theorem.

\begin{observation}
	If a set $T$ of vertices is edge-well-linked in a graph $G$, then $T$ is $1$-well-linked in $G$.
\end{observation}

Next, we state the Boosting Theorems of~\cite{CC14}. 

\begin{theorem}[Theorem 2.14 in~\cite{CC14}]\label{thm: grouping}
	Suppose we are given a connected graph $G=(V,E)$ with maximum vertex degree at most $\Delta\geq 3$ and a set $T\subseteq V$ of $\hat\kappa$ vertices, such that $T$ is $\alpha$-well-linked in $G$, for some $0<\alpha\leq 1$. Then there is a subset $T'\subseteq T$ of $\ceil{\frac{3\alpha\hat\kappa }{10\Delta}}$ vertices, such that $T'$ is node-well-linked in $G$.
\end{theorem}

\begin{theorem}[Theorem 2.9 in~\cite{CC14}]\label{thm: linkedness from node-well-linkedness}
	Suppose we are given a graph $G$ with maximum vertex degree at most $\Delta$, and two disjoint subsets $T_1,T_2$ of vertices of $G$, with $|T_1|,|T_2|\geq \hat\kappa$, such that $T_1\cup T_2$ is $\alpha$-well-linked in $G$, for some $0<\alpha\le 1$, and each one of the sets $T_1,T_2$ is node-well-linked in $G$. Let $T'_1\subseteq T_1$, $T_2'\subseteq T_2$, be any pair of subsets with $|T_1'|=|T_2'|\leq \frac{\alpha\hat\kappa}{2\Delta}$. Then $(T'_1,T_2')$ are linked in $G$.
\end{theorem}

Recall that we are given a weak \PoS $\pos=(\cset,\set{\pset_i}_{i=1}^{\ell-1},A_1,B_{\ell})$ of length $\ell=g^2$ and width $w=g^2$, with $\cset=(C_1,\ldots,C_{\ell})$, in a minor $H''$ of $H$, whose maximum vertex degree is bounded by $4$.
Consider some index $1\leq i<\ell$, and recall that $B_i\subseteq C_i$ and $A_{i+1}\subseteq C_{i+1}$ are the sets of endpoints of the paths of $\pset_i$ lying in $C_i$ and $C_{i+1}$, respectively. Applying Theorem~\ref{thm: grouping} to graph $C_i$ with  $T=B_i$, we obtain a subset $\tilde B_i\subseteq B_i$ of at least $w'=3w/40$ vertices, such that $\tilde B_i$ is node-well-linked in 
$C_i$. Let $\tpset_i\subseteq \pset_i$ be the set of paths originating at the vertices of $\tilde B_i$, and let $\tilde A_{i+1}\subseteq A_{i+1}$ be the set of their endpoints, lying in $C_{i+1}$. We then apply Theorem~\ref{thm: grouping} to graph $C_{i+1}$, with the set $T=\tilde A_{i+1}$ of vertices, to obtain a collection $\tilde A_{i+1}'\subseteq \tilde A_{i+1}$ of at least $w''=3w'/40=\Omega(w)$ vertices, such that $\tilde A_{i+1}'$ is node-well-linked in $C_{i+1}$. Let $\tpset'_i\subseteq \tpset_i$ be the set of paths terminating at the vertices of $\tilde A_{i+1}'$. Finally, we select an arbitrary subset $\pset'_i\subseteq \tpset'_i$ of $\tilde w=\floor{w''/8}=\Omega(w)$ paths, and we let $B'_i\subseteq \tilde{B}_i$ and $A'_{i+1}\subseteq \tilde A'_{i+1}$ be the sets of vertices where the paths of $\pset'_i$ originate and terminate, respectively. 

As our last step, we apply Theorem~\ref{thm: grouping} to graph $C_1$ and a set $T=A_1$ of vertices, to obtain a subset $\tilde A_1\subseteq A_1$ of $w'$ vertices that are node-well-linked in $C_1$, and we select an arbitrary subset $A'_1\subseteq \tilde A_1$ of $\tilde w$ vertices of $A_1$. We select a subset $B'_{\ell}\subseteq B_{\ell}$ of $\tilde w$ vertices similarly.

For each $1\leq i\leq \ell$, we are now guaranteed that each of the sets $A_i',B_i'$ is node-well-linked in $C_i$, and, from Theorem~\ref{thm: linkedness from node-well-linkedness}, $(A_i',B_i')$ are linked in $C_i$. The final strong \PoS is: $\pos'=(\cset,\set{\pset'_i}_{i=1}^{\ell-1},A'_1,B'_{\ell})$; its length is $\ell=g^2$, and its width is $\tilde w=\Omega(g^2)$.

\subsection{Parameters}\label{subsec: params weaker}
For convenience, this subsection summarizes the main parameters used in Section~\ref{sec: Building the Crossbar}.

\begin{itemize}
	\item The cardinalities of the original sets $A,B$ and $X$ of vertices are $\kappa\geq 2^{22}g^{10}\log g$.
	
	\item The depth of the pseud-grid is $D=64g^4$. We obtain $D$ sets $\rset_1,\ldots,\rset_D$ of paths, of cardinality at most $g^2$ each.
	
	\item The total number of paths in the set $\rset$ is $N$, where $D\leq N\leq g^2 D$, so $64g^4\leq N\leq 64g^6$.

	\item The number of slices is $M=8g^4\log g$.
	
	\item For each slice $i$, we define subsets $\Sigma'_i\subseteq \Sigma_i$ and $\qset'_i\subseteq\qset_i$ of paths, such that $(\Sigma'_i,\qset'_i)$ are $(4g^2,D/2)$-intersecting.
	
	\item After the well-linked decomposition, we obtain $M=8g^4\log g$ clusters - one cluster per slice. For each resulting cluster $C_i$, $|\tilde \Sigma_i|\geq D/4$, and the endpoints of the paths in $\tilde\Sigma_i$ are $g^2$-weakly-well-linked in $C_i$.
\end{itemize}

\section{Proof of Theorem~\ref{thm: main for building crossbar}}
\label{sec: stronger structural theorem}

In this section we complete the proof of Theorem~\ref{thm: main for building crossbar}.
Recall that we are given a graph $H$, and three disjoint subsets $A,B,X$ of its vertices, each of cardinality $\kappa\geq 2^{22}g^{9}\log g$, such that every vertex in $X$ has degree $1$.
We are also given a set $\tpset$ of $\kappa$ node-disjoint paths connecting vertices of $A$ to vertices of $B$, and a set $\tqset$ of $\kappa$ node-disjoint paths connecting vertices of $A$ to vertices of $X$. Our goal is to prove that either $H$ contains an $(A,B,X)$-crossbar of width $g^2$, or  that its minor contains a strong Path-of-Sets System whose length and width are both at least $\Omega(g^2)$.
We follow the framework of the proof in Section~\ref{sec: Building the Crossbar}, but we introduce some changes that will allow us to save a factor of $g$ on the parameter $\kappa$. The first three steps of the proof are identical to those in Section~\ref{sec: Building the Crossbar}, except that we use slightly weaker parameters for the slicing.

\subsection*{Step 1: Pseudo-Grid}
We start by defining the sets $\pset,\qset$ of $\kappa$ node-disjoint paths exactly like in Section~\ref{sec: Building the Crossbar}. We set the parameter $D=64g^4$ like in Section~\ref{sec: Building the Crossbar}. Note that the inequality $D\leq \kappa/(2g^2)$ still holds for the new value of $\kappa$, and that Theorem~\ref{thm: pseudo-grid} holds regardless of the specific value of $\kappa$. We apply Theorem~\ref{thm: pseudo-grid} to the graph $H$.
 If the outcome is an $(A,B,X)$-crossbar of width $g^2$, then we return this crossbar and terminate the algorithm. Therefore, we assume from now on that the outcome of the theorem is a pseudo-grid of depth $D$. As before, we denote $N=|\rset|$, so $D\leq N\leq g^2D$ holds.
 
 \subsection*{Step 2: Slicing}
 We define the subset $\qset''\subseteq \qset'$ of paths exactly as before:
 Recall that for each $1\leq i\leq D$, there are at most $2g^2$ paths $Q\in\qset'$, such that $Q$ does not intersect any path of $\rset_i$. We discard all such paths from $\qset'$, obtaining a set $\qset''\subseteq\qset'$ of paths. Therefore, we discard at most $2g^2 D=128g^6<\kappa/8$ paths, and, since $\qset'=\ceil{\kappa/4}$, we get that $|\qset''|\geq \kappa/8$. We are now guaranteed that Property~\ref{prop: intersection} holds (that is, for each $1\leq i\leq D$, each path $Q\in \qset''$ intersects some path of $\rset_i$). 
 As before, we denote by $A'\subseteq A$ and $B'\subseteq B$ the sets of endpoints of the paths of $\rset$ lying in $A$ and $B$, respectively, and we let $H'$  be the sub-graph of $H$, obtained by taking the union of all paths in $\rset$ and all paths in $\qset''$.

We follow the algorithm from Section~\ref{subsec:slicing} to modify the graph $H'$ and the sets $\rset,\qset''$ of paths, such that the resulting graph $H''$ is a minor of $H$; property~\ref{prop: intersection} continues to hold; and $H''$ has the perfect unique linkage property with respect to $A'$ and $B'$, with the unique linkage being $\rset$.

Lastly, we perform a slicing of the paths in $\rset$, by applying Theorem~\ref{thm: slicing} to graph $H''$, but the number of slices that we obtain is somewhat smaller. Specifically, we set  $\hM=128g^3\log g$ and $\hat w=2^{11}g^6$.
Notice that, since $N\leq g^2 D=64g^6$ and $|\qset''|\geq \kappa/8\geq 2^{19}g^9\log g$, $|\qset''| \geq \hat M\hat w+(\hat M+1)\hat N$ holds.
Therefore, we obtain an $M_1$-slicing $\Lambda=\set{\Lambda(R)}_{R\in \rset}$ of $\rset$, of width $w$ with respect to $\qset''$, where $M_1=128g^3\log g$, and $w=2^{11}g^6$. 
As before, for every $1\leq i\leq M_1$, we denote by $\Sigma_i=\set{\sigma_i(R)\mid R\in \rset}$.
We denote the subset $\hqset_i\subseteq \qset''$ of paths corresponding to $\Sigma_i$ by $\qset_i$. We call $(\Sigma_i,\qset_i)$ the \emph{$i$th slice of $\Lambda$}.
Notice that so far we have followed the proof from Section~\ref{sec: Building the Crossbar} exactly, except that the number of slices that we obtain is smaller by factor $g$.

\subsection*{Step 3: Intersecting Path Sets}
This step is also virtually the same as in Section~\ref{sec: Building the Crossbar}.
Let $\hw=4g^2$ and $\hD=D/2$. Consider some $1\leq i\leq M_1$, and the corresponding sets $\Sigma_i,\qset_i$ of paths. Denote $\hrset=\Sigma_i$ and $\hqset=\qset_i$. Recall that each path $Q\in\qset_i$ intersects at least $D$ paths of $\Sigma_i$. It is also easy to verify that $|\hqset|\geq 2|\hrset|\hw\hD$. 
Indeed, $|\hqset|\geq 2^{11}g^6$, while $2|\hrset|\hw /\hD=16N\cdot g^2/D\leq 16g^4$. We can now apply  Lemma~\ref{lemma: intersecting sets} to sets $(\Sigma_i,\qset_i)$ of paths, with parameters  $\hD,\hw$, to obtain a partition $(\Sigma'_i,\Sigma''_i)$ of $\Sigma_i$, and a subset $\qset'_i\subseteq\qset_i$ of paths, such that $(\Sigma'_i,\qset'_i)$ is a $(4g^2,D/2)$-intersecting pair of path sets, $|\qset'_i|\geq 2^{10}g^6$, and each path of $\Sigma_i''$ intersects at most $4g^2$ paths of $\qset_i'$.

So far our proof followed exactly the proof from Section~\ref{sec: Building the Crossbar}. Unfortunately, the number of slices $M_1$ that we obtain is smaller than the one in Section~\ref{sec: Building the Crossbar} by a factor of $g$. In order to get around this problem, we distinguish between two cases. Intuitively, in one of the cases (Case 2) we will be able to increase the number of slices so that we can employ the rest of the proof from Section~\ref{sec: Building the Crossbar}. In the other case (Case 1) the number of slices will stay low, but we will compensate this by ensuring that the cardinalities of the sets $\Sigma_i'$ are large in many of the slices.


For $1\leq i\leq  M_1$, we say that slice $(\Sigma_i,\qset_i)$ is of type $1$ iff $|\Sigma'_i|\geq N/g$, and we say that it is of type $2$ otherwise.
We say that Case $1$ happens if at least half the slices are of type $1$; otherwise, we say that Case $2$ happens.

\subsection{Case 1:  Many Type-$1$ Slices}

If Case $1$ happens, at least half the slices are of type 1. We will ignore all type-2 slices, so we can assume that in Case $1$ the number of slices is $64g^3\log g$, that we denote, abusing the notation, by $M_1$, and that in every slice $(\Sigma_i,\qset_i)$, $|\Sigma'_i|\geq N/g$. Recall that for each $i$, $(\Sigma'_i,\qset'_i)$ are $(4g^2,D/2)$-intersecting, and $|\qset_i'|\geq 2^{10}g^6$.

\subsection*{Step 4: Well-Linked Decomposition}
Fix some index $i$, such that $(\Sigma_i,\qset_i)$ is a type-1 slice. 
As before, we  let $H'_i$ be a graph obtained from the union of the paths in $\Sigma_i$ and $\qset_i$. And we denote $\hat D=D/4$ and $\hat w=g^2$, so that $(\Sigma'_i,\qset'_i)$ are $(4\hw,2\hD)$-intersecting, and $\hat D\geq 8\hat w$, since $D=64g^4$. As in Section~\ref{subsec: WLD}, we apply Theorem~\ref{thm: wld} to graph $H'_i$, with $\hat \Sigma=\Sigma'_i$, $\hat \qset=\qset'_i$, and parameters $\hD$, $\hw$. Let $\cset_i$ be the resulting collection of happy clusters. For each such cluster $C\in \cset_i$, we denote $\tilde\Sigma(C)=\hat\Sigma(C)$. As before, the endpoints of the paths of $\tilde\Sigma(C)$ are $g^2$-weakly well-linked in $C$, and $|\tilde\Sigma(C)|\geq D/4$. We denote $\tilde\Sigma_i=\bigcup_{C\in \cset_i}\tilde\Sigma(C_i)$. Notice that Theorem~\ref{thm: wld} guarantees that $|\tilde \Sigma_i|\geq |\Sigma'_i|/4\geq N/(4g)$.
Let $\cset$ be the union of all sets $\cset_i$, where $i$ is a type-$1$ slice, and let $\tilde\Sigma$ be the union of all corresponding sets $\tilde\Sigma_i$, so $|\tilde \Sigma|\geq M_1 N/4g$.

Intuitively, for some of the clusters $C\in \cset$, $|\tilde\Sigma(C)|$ may be very large, and then it is sufficient for us to have a small number of such clusters. But it is possible that for most clusters in $C$, $|\tilde\Sigma(C)|$ is relatively small -- possibly as small as $D/4$. In this case, it would be helpful for us to argue that the number of such clusters is large. In other words, we would like to obtain a tradeoff between the number of clusters $C$ and the cardinality of their corresponding path sets $|\tilde\Sigma(C)|$. 

In order to do so, we group the clusters geometrically. Recall that for each $C\in \cset$, $D/4\leq |\tilde\Sigma(C)|\leq N\leq g^2 D$. For $0\leq j< 2\log g+2$, we say that cluster $C$ belongs to class $\sset_j$ iff $D\cdot 2^j/4\leq |\tilde\Sigma(C)|<D\cdot 2^{j+1}/4$. If $C\in \sset_j$, then we  say that all paths in $\tilde\Sigma(C)$ belong to class $j$. Then there must be an index $j$, such that the number of paths in $\tilde{\Sigma}$ that belong to class $j$ is at least $\frac{|\tilde\Sigma|}{2\log g+2}\geq \frac{M_1 N}{16g\log g}$. We let $\tilde{\cset}'=\sset_j$, and we let $\tilde \Sigma'\subseteq\tilde \Sigma$ be the set of all paths that belong to class $j$. Note that $|\tilde{\cset}'|\geq \frac{|\tilde\Sigma'|}{D\cdot 2^{j+1}/4}\geq \frac{M_1N}{4D\cdot 2^{j+1}g\log g}=\frac{64g^3N\log g}{8D\cdot 2^{j}g\log g}=\frac{8Ng^2}{D\cdot 2^j}$, while for each $C\in \cset$, $|\tilde\Sigma(C)|\geq D\cdot 2^j/4$.

To summarize the steps in Case $1$ until now, we have obtained a collection of at least $\frac{8Ng^2}{D\cdot 2^j}$ clusters, that we denote, abusing the notation, by $\tilde\cset$. For each slice $i$, we may have a number of clusters of $\tilde \cset$ that belong to that slice. We denote the set of all such clusters by $\tilde\cset_i$. Notice however that, if $C,C'\in \tilde \cset_i$ and $C\neq C'$, then $\tilde\Sigma(C)\cap \tilde\Sigma(C')=\emptyset$. Moreover, if $\rset'\subseteq \rset$ is the set of all paths containing the segments of $\tilde\Sigma(C)$, and $\rset''$ is defined similarly for $C'$, then $\rset'\cap \rset''=\emptyset$. We are guaranteed that for each cluster $C\in \tilde\cset$, $|\tilde\Sigma(C)|\geq D\cdot 2^j/4$, and the endpoints of the paths of $\tilde\Sigma(C)$ are $g^2$-weakly well-linked in $C$.

\subsection*{Step 5: a Path-of-Sets System}

In this step, we employ Theorem~\ref{thm: long intersecting chain} to construct a weak \PoS of length and width $g^2$ in $H''$. We use the parameters $\hat M=|\tilde \cset|\geq\frac{8Ng^2}{D\cdot 2^j}$, $\hw=g^2$, $\hN=N$, and $\hD=\frac{D\cdot 2^j}{4}$. 
Recall that set $\cset$ may contain clusters from the same slice. We order the clusters in $\cset$ as follows: first, we order the clusters in the increasing order of their slices; the clusters inside the same slice are ordered arbitrarily. Let $C_1,\ldots,C_{\hM}$ be the resulting ordering of the clusters. We define the sets $S_1,\ldots,S_{\hM}$ exactly as before. An important observation is that, if $C_i,C_j$ belong to the same slice, then $S_i\cap S_j=\emptyset$. 

We need to verify that the conditions of Theorem~\ref{thm: long intersecting chain} hold for our choice of parameters. The first condition is that $\hN\geq 3\hw$. Since we use the same parameters $\hN$ and $\hw$ as in Section~\ref{subsec:weakPoS}, this inequality continues to hold. The second condition is $\hD^2\geq 4\hat{N}\hw$. Recall that this condition held for the parameters in Section~\ref{sec: Building the Crossbar}, where the values of $\hat{N}$ and $\hw$ were the same, and $\hD=D/4$ was smaller than the current value of $\hD$, so this condition continues to hold. The third condition is that $\hD\hM\geq 2\hat{N}\hw$. This condition is easy to verify by substituting the values of the relevant parameters.

We now use Theorem~\ref{thm: long intersecting chain} to obtain a sequence $1\leq i_1< i_2<\cdots<i_{g^2}\leq \hM$ of indices, such that for all $1\leq j<g^2$, $|S_{i_j}\cap S_{i_{j+1}}|\ge g^2$. Notice that each resulting cluster $C_{i_1},C_{i_2},\ldots,C_{i_{g^2}}$ must belong to a different slice. The remainder of the construction of the weak Path-of-Sets system is done exactly as in subsection~\ref{subsec:weakPoS}. Therefore, we obtain a weak Path-of-Sets system of width $g^2$ and length $g^2$. Our last step is to convert it into a strong Path-of-Sets system using the same procedure as in Section~\ref{subsec: strongPoS}.

\subsection{Case 2:  Many Type-$2$ Slices}

Assume now that Case $2$ happens, so we have at least $\frac{1}{2}M_1=64g^3\log g$ type-2 slices $\{(\Sigma_i,\qset_i)\}_{i}$. Recall that for each $i$, we are given a partition $(\Sigma'_i,\Sigma''_i)$ of $\Sigma_i$, and a subset $\qset'_i\subseteq \qset_i$ of at least $2^{10}g^6$ paths, such that $(\Sigma'_i,\qset'_i)$ are $(4g^2,D/2)$-intersecting, and every path in $\Sigma''_i$ interescts at most $4g^2$ paths of $\qset'_i$. Recall that in each type-$2$ slice $(\Sigma_i,\qset_i)$, $|\Sigma'_i|<N/g$.

 Intuitively, we are interested in either obtaining a large number of slices, or in obtaining a large number of paths in the sets $\Sigma_i'$ of each such slice. Type-$1$ slices achieve the latter. But in type-$2$ slices, the cardinalities of sets $\Sigma_i'$ are too small for us. Fortunately, we can exploit this fact in order to increase the number of slices.

\begin{theorem}\label{thm: slicing for case 2}
	Assume that Case $2$ happens, and that $H$ does not contain an $(A,B,X)$-crossbar of width $g^2$. Then there is an $M_2$-slicing of $\rset$, of width at least $N/g$ with respect to $\qset''$, where $M_2=8g^4\log g$.
\end{theorem}

\begin{proof}
	The proof directly follows from the following lemma.

	\begin{lemma}\label{lem: case 2 slicing of single slice}
		Assume that $H$ does not contain an $(A,B,X)$-crossbar of width $g^2$. Let $1\leq i\leq  M_1$ be an index, such that $(\Sigma_i,\qset_i)$ is a type-$2$ slice of the original slicing $\Lambda$. Then there is an $\hM$-slicing $\Lambda_i$ of $\Sigma_i$ of width at least $N/g$ with respect to $\qset_i$, where $\hM=g$.
	\end{lemma}

Before we provide the proof of Lemma~\ref{lem: case 2 slicing of single slice}, we complete the proof of Theorem~\ref{thm: slicing for case 2} using it.	
Recall that in Case 2, there are at least $64g^3\log g$ slices of $\Lambda$ of type-$2$ slices. Each such slice $(\Sigma_i,\qset_i)$ is then in turn sliced into $g$ slices using Lemma~\ref{lem: case 2 slicing of single slice}. Therefore, by combining the slicing $\Lambda$ together with the individual slicings $\Lambda_i$ for all type-$2$ slices $(\Sigma_i,\qset_i)$, we obtain a slicing of $\rset$, where the number of slices is at least $(64g^3\log g)\cdot g>8g^4\log g=M_2$. The width of the new slicing with respect to $\qset''$ is at least $N/g$.

 From now on we focus on the proof of Lemma~\ref{lem: case 2 slicing of single slice}. Let $1\leq i\leq M_1$ be an index, such that $(\Sigma_i,\qset_i)$ is a type-$2$ slice of the original slicing $\Lambda$. Our goal is to produce an $\hM$-slicing $\Lambda_i$ of $\Sigma_i$ of width at least $N/g$ with respect to $\qset_i$, where $\hM=g$. The idea is that we will discard all paths in $\qset_i\setminus\qset_i'$; ignore the paths in $\Sigma_i''$ (for each such path $\sigma\in \Sigma_i''$ we will eventually produce a trivial slicing, where $v_0(\sigma)$ is the first endpoint of $\sigma$, and $v_1(\sigma)=v_2(\sigma)=\cdots=v_{\hM}(\sigma)$ is its last endpoint), and will focus on slicing the paths of $\Sigma_i'$, trying to achieve a slicing whose width is at least $N/g$ with respect to $\qset_i'$. Unfortunately, some of the paths in $\qset_i'$ may intersect the paths of $\Sigma_i''$, so our first step is to get rid of all such intersections. We do so using the following claim. Recall that $|\qset'_i|\geq 2^{10}g^6$.
	
\begin{claim}\label{claim: type-2 slice, small intersection}
		Assume that $H$ does not contain an $(A,B,X)$-crossbar of width $g^2$, and let $(\Sigma_i,\qset_i)$ be a slice of type 2. Then at least $2^9g^6$ paths in the set $\qset'_i$ are disjoint from the paths in $\Sigma''_i$.
\end{claim}
	
	\begin{proof}
		Assume otherwise. Let $\bset\subseteq \qset'_i$ be the set of all paths that have non-empty intersection with paths in $\Sigma''_i$, so $|\bset|\geq 2^9g^6$. We further partition the set $\bset$ into two subsets: set $\bset_1$ contains all paths $Q\in \bset$ that intersect at least $8g^2$ paths of $\Sigma''_i$, and $\bset_2=\bset\setminus\bset_1$.
		
		We first show that $|\bset_1|\leq 32g^6$. Recall that $|\Sigma''_i|\leq N\leq 64g^6$, and each path $\sigma\in \Sigma''_i$ intersects at most $4g^2$ paths of $\qset'_i$. Therefore, there are at most $256g^8$ pairs $(\sigma,Q)$ of paths, with $\sigma\in \Sigma''_i$ and $Q\in \bset$, such that $\sigma\cap Q\neq \emptyset$. As each path of $\bset_1$ intersects at least $8g^2$ paths of $\Sigma''_i$, we get that $|\bset_1|\leq 256g^8/8g^2\leq 32g^6$.
		
		We conclude that $|\bset_2|\geq 2^8g^6$. We exploit this fact to construct a $(A,B,X)$-crossbar of width $g^2$ in $H$. In order to do so, we perform $g^2$ iterations, where in each iteration we add some path $P$ connecting a vertex of $A$ to a vertex of $B$ to the crossbar, and its corresponding path $Q_P$, thus iteratively constructing the crossbar $(\pset^*,\qset^*)$. In every iteration, we will delete some paths from $\Sigma_i''$ and from $\bset_2$. The path $P$ that we add to $\pset^*$ is a path from $\rset$, that contains some segment of $\Sigma_i''$, and its corresponding path $Q_P$ is a sub-path of a path in $\bset_2$. We start with $\pset^*,\qset^*=\emptyset$, and we maintain the following invariants:
		
		\begin{itemize}
			\item All paths in the current sets $\pset^*,\qset^*$ are disjoint from all paths in $\Sigma''_i,\bset_2$; moreover, each path $R\in \rset$ that contains a segment $\sigma\in \Sigma''_i$ is disjoint from all paths in $\pset^*\cup \qset^*$; and
			\item Each remaining path in $\bset_2$ intersects some path in the remaining set $\Sigma''_i$.
		\end{itemize}
		
		At the beginning, $\pset^*,\qset^*=\emptyset$, and the invariants hold. 
		Assume now that the invariants hold at the beginning of iteration $j$. The $j$th iteration is executed as follows. We let $Q\in \bset_2$ be any path, and we let $\sigma\in \Sigma_i''$ be any path intersecting $Q$. We add the unique path $P\in \rset$ that contains $\sigma$ to $\pset^*$, and we add a sub-path of $Q$, connecting a vertex of $P$ to a vertex of $X$ to $\qset^*$, as $Q_P$ (recall that by the definition of a pseudo-grid, every path in $\qset'$ has an endpoint that lies in $X$)\footnote{In fact we have contracted edges on the paths in $\qset'$ when constructing the graph $H''$ in Step 2, in order to ensure the perfect unique linkage property; formally, in order to obtain the path $Q_P$, we need to un-contract the path $Q$ and then take a segment of the resulting path connecting a vertex of $P$ to a vertex of $X$.}. Let $\sset_j$ be the collection of all paths of the current set $\Sigma_i''$ that intersect $Q$, so $|\sset_j|<8g^2$. Let $\yset_j\subseteq \bset_2$ be the set of all paths that intersect the paths of $\sset_j$, so $|\yset_j|\leq |\sset_j|\cdot 4g^2\leq  32g^4$. We delete the paths of $\sset_j$ from $\Sigma''_i$, and we delete from $\bset_2$ all paths of $\yset_j$. It is easy to verify that the invariants continue to hold after this iteration (for the second invariant, recall that each path of the original set $\bset_2$ intersected some path of $\Sigma''_i$; whenever a path $\sigma$ is deleted from $\Sigma''_i$, we delete all paths that intersect it from $\bset_2$. Therefore, each path that remains in $\bset_2$ must intersect some path that remains in $\Sigma''_i$). In every iteration, at most $32g^4$ paths are deleted from $\bset_2$, while at the beginning $|\bset_2|\geq 2^8g^6$. Therefore, we can carry this process for $g^2$ iterations, after which we obtain an $(A,B,X)$-crossbar of width $g^2$.
	\end{proof}
	
	We let $\tqset_i\subseteq \qset'_i$ be the set of at least $512g^6$ paths of $\qset'_i$ that are disjoint from the paths of $\Sigma''_i$. 
	
	Let $A_i$ be the set of vertices that serve as the first endpoint of the paths in $\Sigma_i$, and let $A'_i\subseteq A_i$ be defined similarly for $\Sigma'_i$. Similarly, let  $B_i$ be the set of vertices that serve as the last endpoint of the paths in $\Sigma_i$, and let $B'_i\subseteq B_i$ be defined similarly for $\Sigma'_i$. We let $H_i$ be the graph obtained from the union of the paths in $\Sigma_i$ and $\qset_i$, and we let $H'_i$ be the  graph obtained from the union of the paths in $\Sigma'_i$ and $\tqset_i$.
	
	\begin{observation}\label{obs: perfect unique linkage}
		Graph $H'_i$ has the perfect unique linkage property with respect to $(A'_i,B'_i)$, with the unique linkage being $\Sigma'_i$.
	\end{observation}
	
	Assume first that the observation is true. Recall that $|\Sigma'_i|\leq N/g$, and $|\tqset_i|\geq 512g^6$. We now invoke Theorem~\ref{thm: slicing} with $\hrset=\Sigma'_i$, $\hqset=\tilde{\qset}_i$, $\hN=|\Sigma'_i|\leq N/g$; $\hM=g$ and $\hw=N/g$, to obtain an $\hM$-slicing of $\Sigma'_i$ of width $N/g$ with respect to $\tqset_i$. In order to do so, we need to verify that $|\tqset_i |\geq \hM \hw+(\hM+1)|\Sigma'_i|$. But $\hM\hw+(\hM+1)|\Sigma'_i|\leq (2\hM+1)N/g\leq 6N$, while $|\tqset_i|\geq 512g^6\geq 6N$, as $N\leq Dg^2=64g^6$. We conclude that there exists an $\hM$ slicing $\Lambda_i$ of $\Sigma'_i$, whose width with respect to $\tqset_i$ is $N/g$. We extend this slicing to an $\hM$-slicing of $\Sigma_i$ in a trivial way: for every path $\sigma\in \Sigma''_i$, we let $v_0(\sigma)$ be its first endpoint, and we let $v_1(\sigma)=v_2(\sigma)=\cdots=v_{\hM}(\sigma)$ be its last endpoint. Since the paths of $\tqset_i$ are disjoint from the paths of $\Sigma''_i$, this defines an $\hM$-slicing of $\Sigma_i$ of width $N/g$ with respect to $\tqset_i$, and hence with respect to $\qset_i$ as well.
	It now remains to prove Observation~\ref{obs: perfect unique linkage}.
	
	\begin{proofof}{Observation~\ref{obs: perfect unique linkage}}
		We first claim that $H_i$ has the unique perfect linkage property with respect to $(A_i,B_i)$, with the unique linkage being $\Sigma_i$. Indeed, for each path $R\in \rset$, no vertex of $R\setminus \sigma_i(R)$ belongs to $H_i$. If there is a different $(A_i,B_i)$ linkage $\tilde\Sigma_i\neq \Sigma_i$, then we could replace $\Sigma_i$ with $\tilde \Sigma_i$ in $\rset$, obtaining a different $(A',B')$-linkage in $H''$, violating Observation~\ref{obs: properties of H''}.
		It is also immediate to verify that every vertex of $H_i$ should lie on some path in $\Sigma_i$, as $\rset$ has the perfect unique linkage property in $H''$.
		
		Consider now the graph $H_i$, after we delete all paths of $\qset_i\setminus \tqset_i$ from it. In this new graph, no path of $\Sigma'_i$ may lie in the same connected component as a path of $\Sigma''_i$, as we have deleted all paths that may intersect the paths of $\Sigma''_i$. Clearly, this new graph still has the perfect unique linkage property with respect to $(A_i,B_i)$, with the unique linkage being $\Sigma_i$. Graph $H'_i$ is obtained from $H_i$ by deleting all paths of $\Sigma''_i$ from it. Since each of these paths lies in a distinct connected component, it is easy to verify that $H'_i$ must have the perfect unique linkage property with respect to $(A'_i,B'_i)$, with the unique linkage being $\Sigma'_i$.
\end{proofof}\end{proof}

If Theorem~\ref{thm: slicing for case 2} returns an $(A,B,X)$-crossbar of width $g^2$, then we output this crossbar and terminate the algorithm. Therefore, we assume from now on that the theorem returns an $M_2$-slicing of $\rset$, whose width with respect to $\qset''$ is at least $N/g$. We will ignore the original slicing for Case 2, and will denote this new slicing by $\Lambda$. Abusing the notation, we denote, for each $1\leq i\leq M_2$, the $i$th slice of this new slicing by $(\Sigma_i,\qset_i)$, where $\qset_i\subseteq \qset''$ and $\Sigma_i$ contains the $i$th segment of each path $R\in \rset$.

As before, for each $1\leq i\leq M_2$, we employ Lemma~\ref{lemma: intersecting sets} in order to partition the set $\Sigma_i$ into two subsets, $\Sigma'_i,\Sigma''_i$, and compute a subset $\qset_i'\subseteq \qset_i$ of paths, such that $(\Sigma'_i,\qset'_i)$ are $(4g^2,D/2)$-intersecting path sets and $|\qset'_i|\geq N/(2g)$. In order to do so, we set $\hat w=4g^2$ and $\hat D=D/2$, and we denote $\hrset=\Sigma_i$ and $\hqset=\qset_i$. 
As before, every path of $\qset_i$ intersects at least $D=2\hat D$ paths of $\Sigma_i$.

We now to verify that $|\hqset|\geq \frac{2|\hrset|\hw}{\hD}$. Indeed, $|\hqset|=|\qset_i|\geq 2|\Sigma_i|\cdot 4g^2/(D/2)=16Ng^2/D$, as $D=64g^4$. We can now use Lemma~\ref{lemma: intersecting sets} to obtain a partition $(\Sigma'_i,\Sigma''_i)$ of $\Sigma$ into two subsets, and a subset $\qset'_i\subseteq\qset_i$ of paths, such that $(\Sigma'_i,\qset'_i)$ are $(4g^2,D/2)$-intersecting path sets and $|\qset'_i|\geq N/(2g)$.

The remainder of the proof repeats the Steps 4--6 from the proof in Section~\ref{sec: Building the Grid}. Observe that our starting point is now identical to the starting point of Step 4 in Section~\ref{sec: Building the Crossbar}.
The values of parameters $D$ and $N$ remained unchanged. We have computed a slicing $\Lambda$ of $\rset$ into $M_2=8g^4\log g$ slices - which is equal to the number of slices used in Section ~\ref{sec: Building the Crossbar}. We have also computed, for each $1\leq i\leq M_2$, subsets $\Sigma'_i\subseteq \Sigma_i$, and $\qset'_i\subseteq \qset_i$ of node-disjoint paths, such that $(\Sigma'_i,\qset'_i)$ are $(4g^2,D/2)$-interesecting. Therefore, we can now repeat Steps 4--6 from the proof in Section~\ref{sec: Building the Grid} to obtain a strong Path-of-Sets system of length $\Omega(g^2)$ and width $\Omega(g^2)$ in $H''$.

\subsection{Parameters} \label{subsec: parameters}

For convenience, this subsection summarizes the main parameters used in Section~\ref{sec: stronger structural theorem}.

\begin{itemize}
\item The cardinalities of the original sets $A,B$ and $X$ of vertices from the statement of Theorem~\ref{thm: main for building crossbar} are $\kappa\geq 2^{22}g^9\log g$.
\item The depth of the pseudo-grid is $D=64g^4$. We obtain $D$ sets $\rset_1,\ldots,\rset_D$ of paths, of cardinality at most $g^2$ each.

\item The total number of paths in the set $\rset$ is $N$, where $D\leq N\leq g^2 D$, so $64g^4\leq N\leq 64g^6$.

\item The number of slices in the first slicing is $M_1=128g^3\log g$. The width of the slicing is $w=2^{11}g^6$. For each type-1 slice $i$, $|\Sigma'_i|\geq N/g$.

\item The number of slices in the second slicing is $M_2=8g^4\log g$, and the width is $w_2=N/g$. 

\item In both Case 1 and Case 2, for each slice $i$, we define subsets $\Sigma'_i\subseteq \Sigma_i$ and $\qset'_i\subseteq\qset_i$ of paths, such that $(\Sigma'_i,\qset'_i)$ are $(4g^2,D/2)$-intersecting.

\end{itemize}

\bibliographystyle{alpha}
\bibliography{REF}

\appendix

\section{Proofs Omitted from Section
	 \ref{sec:prelim}}



\subsection{Proof of Claim~\ref{claim: stitching}}

In order to prove the claim, we only need to show how to construct, for each $1\leq i<\ceil{\ell/2}$, a set $\hat \pset_i$ of $w'$ node-disjoint paths, connecting  the vertices of $B'_{2i-1}$ to the vertices of $A'_{2i+1}$, so that the paths of $\bigcup_{i=1}^{\ceil{\ell/2}}\hat \pset_i$ are disjoint from each other, and they are internally disjoint from $\bigcup_{j=1}^{\ceil{\ell/2}}V(C'_j)=\bigcup_{j=1}^{\ceil{\ell/2}}V(C_{2j-1})$. In order to do so, it is enough to show that for each  $1\leq i<\ceil{\ell/2}$, there is a set $\hat \pset_i$ of $w'$ node-disjoint paths, connecting the vertices of $B'_{2i-1}$ to the vertices of $A'_{2i+1}$, such that the paths of $\hat \pset_i$ are contained in the graph $\pset_{2i-1}\cup C_{2i}\cup \pset_{2i}$.

Fix some $1\leq i<\ceil{\ell/2}$. Let $\hat B=B'_{2i-1}$, and let $\hat A=A'_{2i+1}$. Let $\tilde \pset_1\subseteq \pset_{2i-1}$ be the set of paths originating from the vertices of $\hat B$, and let $\hat X$ be the set of endpoints of the paths of $\tilde \pset_1$ that lie in $C_{2i}$ (see Figure~\ref{fig: stitching}). Similarly, let $\tilde \pset_2\subseteq \pset_{2i}$ be the set of paths terminating at the vertices of $\hat A$, and let $\hat Y$ be the set of endpoints of these paths that lie in $C_{2i}$. Clearly, $\hat X\subseteq A_{2i}$ and $\hat Y\subseteq B_{2i}$. Since $(A_{2i},B_{2i})$ are linked in $C_{2i}$, there is a set $\qset$ of $w'$ node-disjoint paths connecting $\hat X$ to $\hat Y$ in $C_{2i}$. By combining $\tilde \pset_1,\qset,\tilde \pset_2$, we obtain the desired set $\hat \pset_i$ of $w'$ node-disjoint paths connecting $\hat B$ to $\hat A$.

\begin{figure}[h]
	\label{fig: stitching}
	\centering
	\scalebox{0.4}{\includegraphics{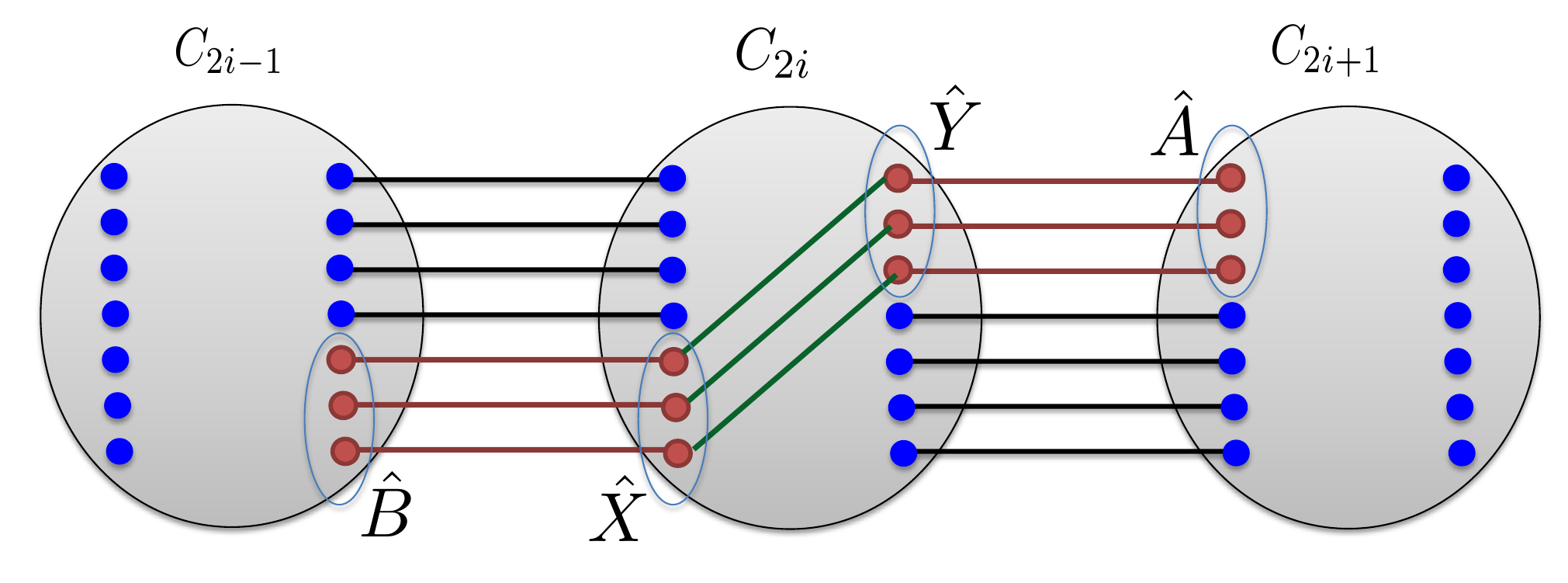}}
	\caption{Stitching in a \PoS. The paths of $\tilde \pset_1$ and $\tilde \pset_2$ are shown in red, the paths of $\qset$ in green. }
\end{figure}

\subsection{Proof of Theorem~\ref{thm: building hairy PoS}}
Our starting point is the following two theorems, that were proved in~\cite{tw-sparsifiers} and~\cite{CC14}, respectively.

\begin{theorem}[Theorem 1.1 in \cite{tw-sparsifiers}]
	\label{thm: degree reduction}
	Let $G$ be a graph of treewidth $k$. Then there is a subgraph $G'$ of $G$, whose maximum vertex degree is $3$, and $\tw(G')=\Omega(k/\poly\log k)$.
\end{theorem}

\begin{theorem}[Theorem 3.4 in \cite{CC14}]
	\label{thm: building PoS}
	There are constants $\hat c,\hat c'>0$, such that for all integers $\ell,w,k>1$ with $k/\log^{\hat c'}k>\hat cw\ell^{48}$, every graph $G$ of treewidth at least $k$ contains a strong \PoS of length $\ell$ and width $w$.
\end{theorem}

Let $G$ be a graph of treewidth at least $k$. We use Theorem~\ref{thm: degree reduction} to obtain a subgraph $G'\subseteq G$ of treewidth  $k'=\Omega(k/\poly \log k)$ and maximum vertex degree $3$.
Let $\ell'=2\ell$ and let $w'=c^*\cdot w$, for a large enough constant $c^*$, that will be determined later. By appropriately setting the values of the constants $c$ and $c'$ in the statement of Theorem~\ref{thm: building hairy PoS}, we can ensure that $k'/\log^{\hat c'}k'>\hat cw'(\ell')^{48}$. From Theorem~\ref{thm: building PoS}, graph $G'$ contains a strong \PoS of length $\ell'$ and width $w'$. Our last step is to turn it into a hairy \PoS of length $\ell$ and width $w$, using the following theorem, that was proved in \cite{chuzhoy2016improved}.

\begin{theorem}[Theorem 6.3 in \cite{chuzhoy2016improved}]
	\label{thm: cluster splitting for hairy PoS}
	For every integer $\Delta>0$, there is an integer $c_{\Delta}>0$ depending only on $\Delta$, such that the following holds. Let $G$ be any graph of maximum vertex degree at most $\Delta$, and let $A,B$ be two disjoint subsets of vertices of $G$, with $|A|=|B|=\kappa$, such that $A$ and $B$ are each node-well-linked in $G$, and $(A,B)$ are node-linked in $G$. Then there are two disjoint clusters $C',S'\subseteq G$, a set $\qset$ of at least $\kappa/c_{\Delta}$ node-disjoint paths connecting vertices of $C'$ to vertices of $S'$, so that the paths of $\qset$ are internally disjoint from $C'\cup S'$, and two subsets $A'\subseteq A\cap C'$, $B'\subseteq B\cap C'$ of at least $\kappa/c_{\Delta}$ vertices each such that, if we denote by $X'$ and $Y'$ the endpoint of the paths of $\qset$ lying in $C'$ and $S'$ respectively, then:
	
	\begin{itemize}
		\item set $Y'$ is node-well-linked in $S'$;
		
		\item each of the three sets $A',B'$ and $X'$ is node-well-linked in $C'$; 
	%
		and
		
		\item every pair of sets in $\set{A',B',X'}$ is node-linked in $C'$.
	\end{itemize}
\end{theorem}

Using the above theorem, we show that any \PoS can be transformed into a hairy \PoS of roughly the same length and width, in the following lemma.

\begin{lemma}
	\label{lem:PoS to hairy PoS}
	Let $G'$ be a graph of maximum vertex degree $3$, and assume that for some $\ell',w'>0$, $G$ contains a strong \PoS of length $\ell'$ and width $w'$. Then $G$ contains a hairy \PoS with length at least $\ell'/2$ and width at least $w'/(3c_{\Delta})$, where $c_{\Delta}$ is the constant from Theorem~\ref{thm: cluster splitting for hairy PoS}.
\end{lemma}

Setting the constant $c^*$ from the definition of $w'$ to be $3c_{\Delta}$, from the above lemma, graph $G'$ contains a hairy \PoS of length at least $\ell=\ell'/2$ and width at least $w=w'/(3c_{\Delta})$, completing the proof of Theorem~\ref{thm: building hairy PoS}. It now remains to prove Lemma~\ref{lem:PoS to hairy PoS}.

\begin{proofof}{Lemma~\ref{lem:PoS to hairy PoS}}
	Let $\pos=(\sset,\set{\pset_i}_{i=1}^{\ell'-1},A_1,B_{\ell'})$ be the given strong \PoS in $G'$, of length $\ell'$ and width $w'$. Recall that the maximum vertex degree in $G'$ is $3$. Let $1\leq i\leq \ell'$ be an odd integer. We apply Theorem~\ref{thm: cluster splitting for hairy PoS} to graph $C_i$, with $A=A_i$ and $B=B_i$. We denote the resulting two clusters $C'$ and $S'$ by $C'_i$ and $S'_i$, respectively, and we denote the resulting subsets $A',B',X',Y'$ of vertices by $A''_i,B''_i,X''_i$, and $Y''_i$, respectively (recall that the cardinality of each such vertex set is at least $w'/c_{\Delta}$). We also denote the corresponding set $\qset'$ of paths by $\qset'_i$.
	One difficulty is that we are not guaranteed that $X''_i$ is disjoint from $A''_i\cup B''_i$. But it is easy to verify that we  can select subsetes $A'_i\subseteq A''_i,B'_i\subseteq B''_i,X'_i\subseteq X''_i$ of cardinalities $\ceil{w'/(3c_{\Delta})}$ each, such that all three sets $A'_i,B'_i,X'_i$ are disjoint. We let $\qset_i\subseteq\qset'_i$ be the set of paths originating at the vertices of $X'_i$, and we let $Y'_i$ be the set of their endpoints that belong to $S_i$.

	Notice that for each odd integer $1\leq i\leq \ell'$, we have now selected two subsets $A'_i\subseteq A_i$ and $B'_i\subseteq B_i$ of $\ceil{w'/(3c_{\Delta})}$ vertices. Using Claim~\ref{claim: stitching}, we can construct a new \PoS $\pos'= (\cset',\set{\pset'_i}_{i=1}^{\ceil{\ell'/2}},A'_1,B'_{\ceil{\ell'/2}})$ of length $\ceil{\ell'/2}$ and width $\ceil{w'/(3c_{\Delta})}$, such that $\cset'=(C_1,C_3,\ldots,C_{2\ceil{\ell'/2}-1})$, and for each $1\leq i<\ceil{\ell'/2}$, the paths in $\pset'_i$ connect the vertices of $B'_{2i-1}$ to the vertices of $A'_{2i+1}$. 
	
	For each $1\leq i\leq \ceil{\ell'/2}$, we replace cluster $C_{2i-1}$ with $C'_{2i-1},\qset_{2i-1}$, and $S_{2i-1}$, to obtain a hairy \PoS of length $\ell'/2$ and width $\ceil{w'/(3c_{\Delta})}$.
\end{proofof}

\section{Proof of Lemma~\ref{lemma: numbering}}
The proof presented here is almost identical to the proof of~\cite{robertson1983graph}. 
We construct a directed graph $G'$ over the set $V(\hat G)$ of vertices, and prove that this graph is a directed acyclic graph. We will then use the natural ordering of the vertices of $V(G')$ in order to define the bijection $\mu$.

We now define the graph $G'=(V',E')$. The set of vertices of $G'$ is $V'=V(\hat G)$. The set of edges consists of two subsets. First, for every path $R\in \hrset$, for every pair $(v,v')$ of distinct vertices of $R$, such that $v$ lies before $v'$ on $R$, we add a directed edge $(v,v')$ to $G'$. The resulting edges are called type-$1$ edges, and the set of these edges is denoted by $E_1$. We now define the second set $E_2$ of edges. Let $R,R'$ be two {\bf distinct} paths of $\hrset$, and let $v\in V(R)$, $v'\in V(R')$ be a pair of vertices. We add a directed edge $(v,v')$ to $E_2$ iff there is a vertex $u\in V(R)$, that lies strictly after $v$ on $R$, and there is an edge $(u,v')$ in $\hat G$. The resulting edge $(v,v')$ is called a type-2 edge, and we say that $u$ is the \emph{witness} for this edge. We then let $E'=E_1\cup E_2$, completing the definition of the graph $G'$.
The following observation is immediate from the definition of the edges of $G'$.

\begin{observation}\label{obs: transitivity of edges}
	Let $R\in \hrset$ be a path, and let $v,v'$ be two vertices of $R$, such that $v$ lies strictly before $v'$ on $R$. Let $u\in V'$ be any vertex, such that $(v',u)\in E'$. Then $(v,u)\in E'$.
\end{observation}
\begin{proof}
	If $u$ also lies on $R$, then $u$ appears after $v'$ on $R$, and hence it also appears after $v$ on $R$, so $(v,u)\in E_1$. Otherwise, if $w$ is the witness for the edge $(v',u)$, then $w$ appears after $v'$ on $R$, and hence $(v,u)\in E_2$, with the witness $w$.
\end{proof}

\begin{claim}\label{claim: DAG}
	Graph $G'$ is a directed acyclic graph.
\end{claim}

\begin{proof}
	Assume for contradiction that $G'$ is not acyclic, and let $C=(v_1,v_2,\ldots,v_r)$ be a directed cycle in $G'$, containing fewest vertices  (we allow a cycle to consist of two vertices). 
	
	\begin{observation}
		The vertices $v_1,\ldots,v_r$ lie on distinct paths of $\hrset$.
	\end{observation}

\begin{proof}
	Assume otherwise, so there are indices $1\leq i<j\leq r$, such that $v_i$ and $v_j$ lie on the same path of $\hrset$. Denote this path by $R$.
	In this proof, the addition in the subscripts is modulo $r$, that is, for $z=r$, $v_{z+1}=v_1$, and for $z=1$, $v_{z-1}=v_r$. 
	
	 Assume first that $v_i$ appears before $v_j$ on the path $R$. Then, from Observation~\ref{obs: transitivity of edges}, edge $(v_i,v_{j+1})$ belongs to $E'$. Therefore, $v_1,\ldots,v_i,v_{j+1},\ldots,v_r$ is a cycle in $G'$, containing fewer vertices that $C$ (if $j=r$, then we use the cycle $(v_1,\ldots,v_i)$), contradicting the minimality of the cycle $C$.
	
	Notice that it is impossible that $v_i=v_j$: since graph $G'$ does not contain loops, if $v_i=v_j$, then $i$ and $j$ cannot be consecutive indices, and so $(v_1,\ldots,v_i,v_{j+1},\ldots,v_r)$ is a cycle in $G'$ containing fewer vertices than $C$, a contradiction. 
	
	Therefore, $v_i$ appears strictly after $v_j$ on $R$. In this case, $(v_i,v_j)$ cannot be an edge of $G'$, and so $j\neq i+1$.  But then, from  Observation~\ref{obs: transitivity of edges}, edge $(v_j,v_{i+1})$ belongs to $E'$, and so $(v_{i+1},\ldots,v_j)$ is a cycle in $G'$, containing fewer vertices than $C$, contradicting the minimality of $C$.
\end{proof}

We conclude that every vertex of $C$ appears on a distinct path of $\hrset$. For each $1\leq i\leq r$, let $R_i$ be the path of $\hrset$ containing the vertex $v_i$, and let $u_i$ be the witness for the edge $(v_i,v_{i+1})$. We let $\rset'=\set{R_1,\ldots,R_r}$ and $\rset''=\hrset\setminus \rset'$. For each $1\leq i\leq r$, let $a_i$ and $b_i$ be the first and the last endpoints of the path $R_i$, respectively, and let $A'=\set{a_1,\ldots,a_r}$, and $B'=\set{b_1,\ldots,b_r}$. We construct an $(A',B')$-linkage $\tilde \rset$ in $\hG$, that is different from $\rset'$, and disjoint from $\rset''$. This implies that $\tilde \rset\cup \rset''$ is an $(\hat A,\hat B)$-linkage in $\hG$, that is different from $\hrset$, contradicting the uniqueness of $\hrset$.

The linkage $\tilde \rset$ consists of $r$ paths $R'_1,\ldots,R'_r$, where for all $1\leq i\leq r$, path $R'_i$ starts at $a_i$ and terminates at $b_{i-1}$ (path $R'_1$ starts at $a_1$ and terminates at $b_r$).

For $1\leq i\leq r$, path $R'_i$ consists of: (i) the segment of $R_i$ from $a_i$ to $v_i$; (ii) the edge $(v_i,u_{i-1})$ of $\hG$ (for $i=1$, we use edge $(v_1,u_r)$); and (iii) the segment of $R_{i-1}$ from $u_{i-1}$ to $b_{i-1}$ (see Figure~\ref{fig: rerouting}).

\begin{figure}[h]
	\centering
	\scalebox{0.3}{\includegraphics{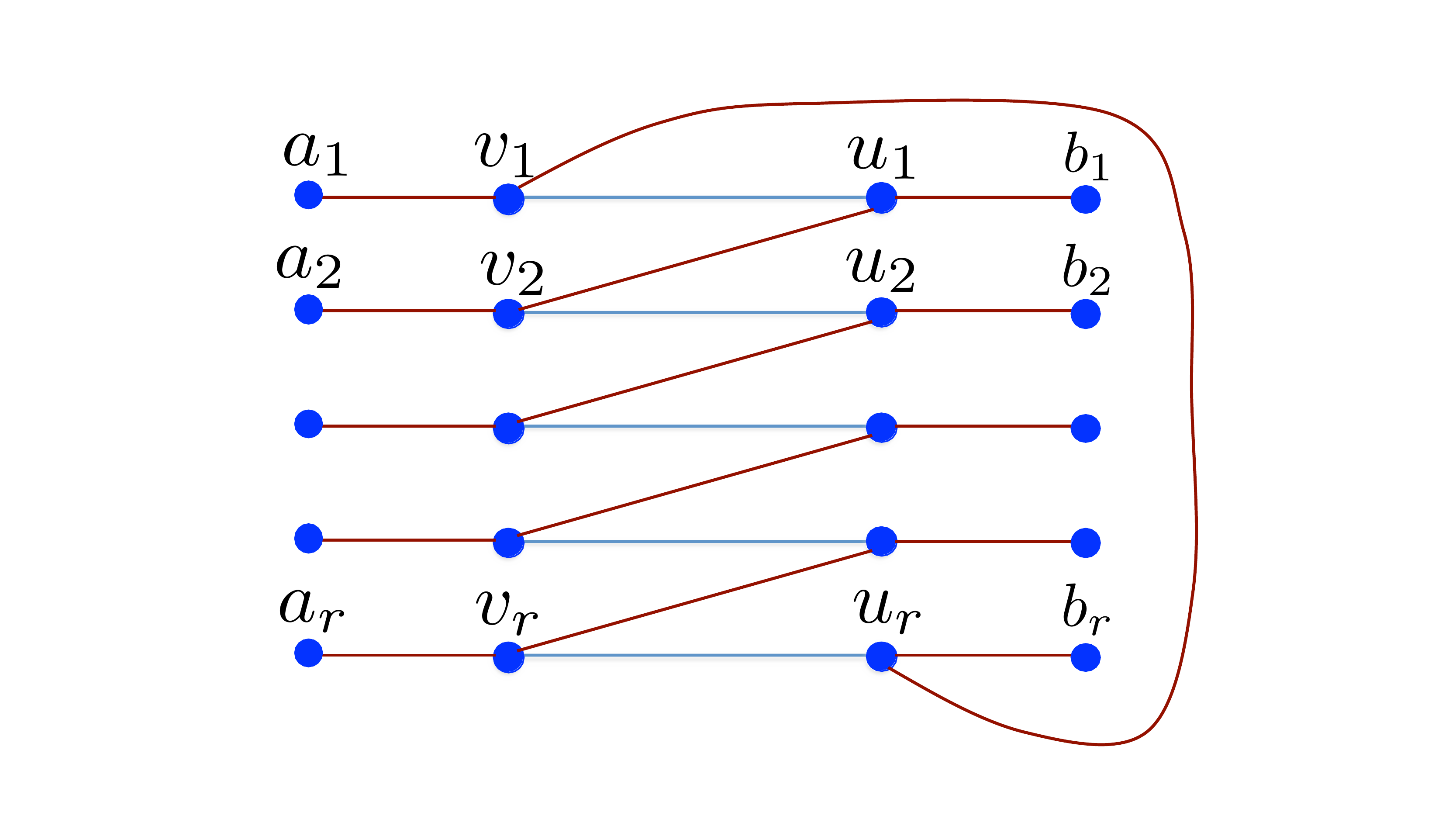}}
	\caption{Rerouting the paths in $\rset'$. The new set $\tilde \rset$ of paths is shown in red. \label{fig: rerouting}}
\end{figure}
\end{proof}

For convenience, let $\hat n=|V(\hat G)|$.
We now define the bijection $\mu$ of the vertices of $V(\hat G)$ to $\set{1,\ldots,\hat n}$, in a natural way, using $G'$. We perform $\hat n$ iterations, where in iteration $i$ we select a vertex $v_i\in V(G')$, setting $\mu(v_i)=i$. The first iteration is executed as follows. Since $G'$ is a directed acyclic graph, there is some vertex $v\in V'$ with no incoming edges. We denote $v_1=v$, and set $\mu(v)=1$. Consider now some iteration $i>1$, and assume that vertices $v_1,\ldots,v_{i-1}$ were already defined. Consider the graph $G'_i=G'\setminus\set{v_1,\ldots,v_{i-1}}$. This graph is again a directed acyclic graph, so it contains some vertex $v$ with no incoming edges. We let $v_i=v$, and we set $\mu(v_i)=i$. This finishes the definition of the mapping $\mu$. We now verify that it has the required properties.

For the first property, if $R\in \hrset$ is any path, and $v,v'$ are distinct vertices of $R$, with $v'$ appearing before $v$ on $R$, then the edge $(v',v)$ is present in $G'$, and so $\mu(v')<\mu(v)$.

Consider now some integer $1\leq t\leq \hat n$. Recall that set $S_t$ contains, for every path $R\in \hrset$, the first vertex $v$ on $R$ with $\mu(v)\geq t$; if no such vertex exists, then we add the last vertex of $R$ to $S_t$. Recall also that $Y_t=\set{v\in V(\hat G)\mid \mu(v)<t}$ and $Z_t=\set{v\in V(\hat G)\mid \mu(v)\geq t}$. We need to show that graph $\hat G\setminus S_t$ contains no path connecting a vertex of $Y_t$ to a vertex of $Z_t$. Assume otherwise. Then there must be some pair of vertices $v\in Y_t\setminus S_t$ and $v'\in Z_t\setminus S_t$, such that $(v,v')\in E(\hat G)$.
Let $R\in \hrset$ be the path containing $v$, and let $R'\in \hrset$ be the path containing $v'$.

Assume first that $R=R'$. Since $v\in  Y_t\setminus S_t$, $\mu(v)<t$, and, since $v'\in  Z_t\setminus S_t$, $\mu(v')\geq t$. Moreover, there must be some other vertex $u$ of $R$ that belongs to $S_t$, and, from the definition of $S_t$, $u$ appears between $v$ and $v'$ on $R$. Let $a,b$ be the endpoints of the path $R$. Then we can construct a new path $R'$, connecting $a$ to $b$, as follows: path $R'$ consists of the segment of $R$ from $a$ to $v$; the edge $(v,v')$, and the segment of $R$ from $v'$ to $b$. But then, replacing $R$ with $R'$ in $\hrset$, we obtain an $(\hat A,\hat B)$-linkage in $\hG$ that is different from $\hrset$,  contradicting the unique linkage property of $\hat G$.

Therefore, $R\neq R'$ must hold. Let $x$ be the vertex of $R'$ that belongs to $S_t$. From the definition of $S_t$, $\mu(x)\geq t$, and $x$ appears before $v'$ on path $R'$. But then there must be an edge $(x,v)$ in $E_2\subseteq E(G')$, whose witness is $v'$. From the definition of $\mu$, if such an edge is present in $G'$, then $\mu(x)<\mu(v)$ must hold. But $\mu(x)\geq t$, while $\mu(v)<t$, a contradiction.

\end{document}